\PassOptionsToPackage{final}{graphicx}
\PassOptionsToPackage{nosumlimits,nonamelimits}{amsmath}
\PassOptionsToPackage{colorlinks,linkcolor={blue},citecolor={blue},urlcolor={red},breaklinks=true,final}{hyperref}
\documentclass[sigconf,screen,
final,9pt,nonacm]{acmart}
\usepackage[T1]{fontenc}
\usepackage{amsmath}
\usepackage{stmaryrd}
\usepackage{tikz-cd}
\usepackage{stmaryrd}[only,leftrightarroweq,llbracket,rrbracket]
\usepackage[capitalise,nameinlink]{cleveref}
\Crefname{thm}{Theorem}{Theorems}
\usepackage{enumitem}
\usepackage{proof}
\usepackage{stackrel}
\usepackage{mathtools}

\usepackage[layout=footnote,marginclue,final]{fixme}
\FXRegisterAuthor{bk}{abk}{BK}	
\FXRegisterAuthor{km}{akm}{KM}	
\FXRegisterAuthor{ls}{als}{LS}	
\FXRegisterAuthor{jf}{ajf}{JF}	
\FXRegisterAuthor{pw}{apw}{PW}	
\FXRegisterAuthor{pn}{apn}{PN}	

\usepackage[appendix=inline, bibliography=common]{apxproof}

\theoremstyle{plain}
\newtheorem{thm}{Theorem}[section]
\newtheorem{lem}[thm]{Lemma}
\newtheorem{propn}[thm]{Proposition}
\newtheorem{cor}[thm]{Corollary}

\theoremstyle{definition}
\newtheorem{defn}[thm]{Definition}
\newtheorem{expl}[thm]{Example}
\newtheorem{rem}[thm]{Remark}

\newtheorem{alg}{Algorithm}

\usepackage[colorinlistoftodos,textsize=tiny,color=orange!70
]{todonotes}

\newcommand{\frel}[2]{#1 \mathbin{\ooalign{$\rightarrow$\cr$\hspace{0.4ex}+$\cr}_{\V}} #2}
\newcommand{\tworel}[2]{#1 \mathbin{\ooalign{$\rightarrow$\cr$\hspace{0.4ex}+$\cr}} #2}
\newcommand{\into}{\hookrightarrow}
\newcommand{\rev}[1]{#1^{\circ}}
\newcommand{\sfun}{\mathcal{S}}
\newcommand{\dfun}{\mathcal{D}}
\newcommand{\pfun}{\mathcal{P}}
\newcommand{\cfun}{\mathcal{C}}
\newcommand{\ALC}{\mathcal{ALC}}
\newcommand{\V}{\mathcal{V}}
\newcommand{\Lang}{\mathcal{L}}
\newcommand{\FLang}{\mathcal{F}}
\newcommand{\Act}{\mathcal{A}}
\newcommand{\Rat}{\mathbb{Q}}
\newcommand{\Sem}[1]{\llbracket #1 \rrbracket}
\newcommand{\FA}{\mathfrak{A}}
\newcommand{\gen}[1]{\mathsf S^{#1}}
\newcommand{\Gen}[1]{\mathsf S[#1]}
\newcommand{\laxtwo}{L}
\newcommand{\laxv}{{\mathbf L}}
\newcommand{\kantv}{{\mathbf K}}
\newcommand{\logdist}{{d^{\mathsf{log}}}}
\newcommand{\disteps}[1]{d_{#1}}
\newcommand{\Op}{{\mathsf{op}}}

\newcommand{\set}{\mathbf{Set}}

\renewcommand{\epsilon}{\varepsilon}
\renewcommand{\phi}{\varphi}

\def\Xint#1{\mathchoice
	{\XXint\displaystyle\textstyle{#1}}
	{\XXint\textstyle\scriptstyle{#1}}
	{\XXint\scriptstyle\scriptscriptstyle{#1}}
	{\XXint\scriptscriptstyle\scriptscriptstyle{#1}}
	\!\int}
\def\XXint#1#2#3{{\setbox0=\hbox{$#1{#2#3}{\int}$}
		\vcenter{\hbox{$#2#3$}}\kern-.5\wd0}}
\def\dashint{\Xint-}

\title[Computing Distinguishing Formulae for Threshold-Based
  Behavioural Distances]{Computing Distinguishing Formulae\\ for Threshold-Based
  Behavioural Distances}
\author{Jonas Forster}
\orcid{0000-0002-5050-2565}

\affiliation{
  \institution{Friedrich-Alexander-Universität Erlangen-Nürnberg}
  \country{Germany}}

\author{Lutz Schr\"{o}der}
  
\orcid{0000-0002-3146-5906}             
\affiliation{
  \institution{Friedrich-Alexander-Universit\"{a}t Erlangen-N\"{u}rnberg}            
  \country{Germany}                    
}

\author{Paul Wild}
\orcid{0000-0001-9796-9675}
\affiliation{
  \institution{Friedrich-Alexander-Universität Erlangen-Nürnberg}
  \country{Germany}}
 
\author{Barbara K{\"o}nig}
\orcid{0000-0002-4193-2889}
\affiliation{
  \institution{Universität Duisburg-Essen}
\country{Germany}}

\author{Pedro Nora}
\orcid{0000-0001-8581-0675}

\affiliation{
  \institution{Universität Duisburg-Essen}
  \country{Germany}}

\begin{document}

\begin{abstract}
  Behavioural distances generally offer more fine-grained means of
  comparing quantitative systems than two-valued behavioural
  equivalences. They often relate to quantitative modalities, which
  generate quantitative modal logics that characterize a given
  behavioural distance in terms of the induced logical distance. We
  develop a unified framework for behavioural distances and logics
  induced by a special type of modalities that lift two-valued
  predicates to quantitative predicates. A typical example is the
  probability operator, which maps a two-valued predicate~$A$ to a
  quantitative predicate on probability distributions assigning to
  each distribution the respective probability
  of~$A$. Correspondingly, the prototypical example of our framework
  is $\epsilon$-bisimulation distance of Markov chains, which has
  recently been shown to coincide with the behavioural distance
  induced by the popular Lévy-Prokhorov distance on
  distributions. Other examples include behavioural distance on metric
  transition systems and Hausdorff behavioural distance on fuzzy
  transition systems. Our main generic results concern the
  polynomial-time extraction of distinguishing formulae in two
  characteristic modal logics: A two-valued logic with a notion of
  satisfaction up to~$\epsilon$, and a quantitative modal logic. These
  results instantiate to new results in many of the mentioned
  examples. Notably, we obtain polynomial-time extraction of
  distinguishing formulae for $\epsilon$-bisimulation distance of
  Markov chains in a quantitative logic featuring a `generally'
  modality used in probabilistic knowledge representation.
\end{abstract}

\maketitle

\section{Introduction}
\label{sec:introduction}

In systems carrying quantitative information, behavioural distances
often provide a more stable and fine-grained means of comparing the
behaviour of processes than two-valued
equivalences~\cite{GiacaloneEA90}. Behavioural distances can be
treated parametrically over the system type (non-deterministic,
weighted, probabilistic, neighbourhood-based etc.) in the framework of
universal coalgebra~\cite{Rutten00}, in which the system type is
encapsulated by the choice of a functor on a suitable base category,
often the category $\set$ of sets and maps. For instance, Markov
chains are coalgebras for the discrete distribution functor on
$\set$. Coalgebraic treatments of behavioural distances can be based
either on liftings of the given functor to the category of metric
spaces~\cite{bbkk:coalgebraic-behavioral-metrics} or on so-called
(quantitative) \emph{lax extensions} of the functor, which extend the
functor to act on quantitative relations~\cite{WildSchroder22}. Both
lax extensions and functor liftings can be induced from a choice of
quantitative modalities, interpreted coalgebraically as quantitative
predicate
liftings~\cite{WildSchroder22,bbkk:coalgebraic-behavioral-metrics}. A
well-known example is the standard Kantorovich distance on probability
distributions, which is induced by the expectation modality; that is,
given a metric on the outcome space, the distance of two probability
distributions is defined as the supremum of the deviations between the
respective expected values, taken over all non-expansive predicates on the
outcome space.

In the present work, we develop a coalgebraic framework for
behavioural distances that are based on fixing threshold values for
allowed deviations in long-term behaviour. One central example is the
behavioural distance on labelled Markov chains induced by
$\epsilon$-bisimilarity~\cite{DesharnaisEA08}: One defines a notion of
$\epsilon$-bisimulation, which, roughly speaking, works like the
two-valued notion of probabilistic bisimulation but allows
probabilities to deviate by up to~$\epsilon$; then, the induced
behavioural distance between states $x,y$ is the infimum over
all~$\epsilon$ such that $x,y$ are $\epsilon$-bisimilar. It was
recently shown that this behavioural distance, originally called
\emph{$\epsilon$-distance}, is induced by the L\'evy-Prokhorov metric
on distributions~\cite{DesharnaisSokolova25}; we will hence refer to
it as \emph{L\'evy-Prokhorov behavioural distance}.  Due to its
favourable statistical robustness properties, the L\'evy-Prokhorov
distance is popular in machine learning tasks such as conformal
prediction~\cite{AolariteiEA25} and corruption
resistance~\cite{BennounaEA23}; distances related to
$\epsilon$-distance have been employed, for instance, in mobile
security~\cite{DiniEA13} and in differential
privacy~\cite{BartheEA12}.

Our coalgebraic treatment of similar distances, which we term
\emph{threshold-based}, starts from the choice of a set of modalities
of particular type: They are induced by predicate liftings that turn
$2$-valued predicates on the base set into $\V$-valued predicates on the
functorial image of this set, where we generally write~$\V$ for the
unit interval $[0,1]$; we refer to such predicate liftings as
\emph{$2$-to-$\V$}. For instance, L\'evy-Prokhorov behavioural
distance is induced from the single $2$-to-$\V$ predicate lifting that
just evaluates probability distributions on $2$-valued
predicates. Beyond this, the framework subsumes many known examples of
coalgebraic behavioural distance, such as Hausdorff behavioural
distances on metric~\cite{afs:linear-branching-metrics} or
fuzzy~\cite{WildSchroder22} transition systems; the most notable
non-example are Kantorovich-type distances on probabilistic
systems.

Throughout, we cover both symmetric distances, i.e.~behavioural
pseudometrics, and asymmetric ones, i.e.~behavioural hemimetrics such
as simulation distance on metric transition
systems~\cite{FahrenbergLegay14}. Correspondingly, we work with a
notion of \emph{$\epsilon$-simulations}, which take on the nature of
bisimulations if the underlying set of modalities is closed under
duals. Our main interest is then in the algorithmic construction of
\emph{distinguishing} formulae in dedicated modal logics,
i.e.~formulae witnessing high behavioural distance of states.  Our
main results on this framework are the following:
\begin{enumerate}[wide]
\item\label{contrib:lax} We show that the behavioural distance induced
  from a notion of $\epsilon$-(bi-)similarity coincides with the one
  induced by a \emph{threshold-based lax extension} induced by the
  given modalities (\Cref{prop:lp-epsilon-sim}). In the probabilistic
  case, this lax extension is the L\'evy-Prokhorov extension, so one
  instance of this result is the above mentioned characterization of
  $\epsilon$-distance as L\'evy-Prokhorov behavioural
  distance~\cite{DesharnaisSokolova25}.
\item From the underlying $2$-to-$\V$ predicate liftings, we construct
  quantitative predicate
  liftings~\cite{SchroderPattinson11,WildSchroder22}, which lift
  $\V$-valued predicates to $\V$-valued predicates, by Sugeno
  integration~\cite{SugenoThesis} (\Cref{sec:quant-log}). We show that
  the lax extension induced by these quantitative \emph{Sugeno
    modalities} coincides with the threshold-based lax extension from
  Item~\eqref{contrib:lax} (\Cref{thm:lp-kantorovich-coalg}). It
  follows from general
  results~\cite{KonigMikaMichalski18,WildSchroder22} that the
  quantitative modal logic of the Sugeno modalities is
  \emph{characteristic} for threshold-based behavioural distance on
  finitely branching systems in the sense that logical distance
  coincides with behavioural distance. In the probabilistic case, the
  Sugeno modality is precisely the \emph{generally} modality that has
  been used in probabilistic knowledge
  representation~\cite{SchroderPattinson11}; thus, one instance of our
  generic result is a recent result showing that the \emph{generally}
  modality induces the L\'evy-Prokhorov extension~\cite{WildEA25}.
\item We introduce a two-valued modal logic equipped with a notion of
  satisfaction up to~$\epsilon$, which we show characterizes
  threshold-based behavioural distance on finitely branching systems
  (\Cref{sec:two-valued-logic}). The probabilistic instance yields
  Desharnais et al.'s characteristic logic for L\'evy-Prokhorov
  behavioural distance~\cite{DesharnaisEA08}, restricted to finitely
  branching systems. All other instances appear to be new; these
  include a two-valued characteristic logic for behavioural distance
  on metric transition systems~\cite{afs:linear-branching-metrics}.
\item For both the two-valued and the quantitative modal logic, we
  provide an algorithm that extracts distinguishing formulae
  witnessing the failure of two given states to be $\epsilon$-similar
  from Spoiler strategies in a codensity-style~\cite{KomoridaEA19}
  Spoiler-Duplicator game. The extracted formulae are of polynomial
  dag size (the tree size of distinguishing formulae is worst-case
  exponential even in standard two-valued Kripkean modal logics,
  e.g.~\cite{FigueiraGorin10}), and under mild assumptions on the
  involved modalities are computable in polynomial time
  (\Cref{thm:extract-two-valued,thm:extract-quant}). This recovers a
  corresponding result for metric transition
  systems~\cite{afs:linear-branching-metrics}. All other instances of
  the generic result appear to be new; notably, we obtain
  polynomial-time extraction of characteristic formulae, both
  two-valued and quantitative, for L\'evy-Prokhorov distance of Markov
  chains.
\end{enumerate}

\subsubsection*{Related Work} As mentioned above, Desharnais et
al.~\cite{DesharnaisEA08} introduce $\epsilon$-distance on labelled
Markov chains and provide an associated two-valued characteristic
modal logic, without however considering size estimates (indeed, the
proof of the corresponding result~\cite[Theorem~3]{DesharnaisEA08} is
based on a continuity argument and begins by fixing an enumeration of
all formulae, so is not likely to yield a tight size estimate). The
fact that the quantitative modal logic of \emph{generally} is
characteristic for L\'evy-Prokhorov behavioural distance follows from
mentioned recent results by Wild et al.~\cite{WildEA25}. Again,
however, our size estimate and polynomial-time extraction algorithm
for distinguishing formulae are new and do not follow from the
previous result, which relies on fairly involved general constructions
that approximate non-expansive properties at every modal
depth~\cite{KonigMikaMichalski18,WildSchroder22}.

Our generic algorithm for the computation of distinguishing formulae
takes inspiration from a corresponding algorithm for the modal logic
of metric transition systems~\cite{afs:linear-branching-metrics},
which as indicated above is at the same time the only previously known
instance of our general result. The computation of distinguishing
formulae for two-valued behavioural equivalences goes back to work on
labelled transition systems~\cite{c:automatically-explaining-bisim};
by now, corresponding algorithms have been designed in coalgebraic
generality~\cite{kms:non-bisimilarity-coalgebraic,WissmannEA22}. In
the setting of behavioural equivalences, the computation of
distinguishing formulae can be performed as an extension to highly
efficient partition refinement
algorithms~\cite{c:automatically-explaining-bisim,WissmannEA22}. This
leads to very stringent quasi-linear time bounds~\cite{WissmannEA22}
that presumably do not generalize to the quantitative case.

As mentioned above, one case not covered by our framework of
threshold-based behavioural distances are Kantorovich-type
distances on Markov chains, for which an algorithm computing
distinguishing formulae has been presented
recently~\cite{rb:explainability-labelled-mc}. The complexity estimate
of this algorithm is incomparable to ours: While our algorithm finds a
formula witnessing distance strictly above a given~$\epsilon$ in
polynomial time, the algorithm in \emph{op.~cit.} computes, for
given~$n$, a formula serving as an exact witness of the distance under
the $n$-th iterate of the functional defining the behavioural distance
as a fixpoint, in polynomial time.

%

\section{Preliminaries}
\label{sec:preliminaries}

We assume basic familiarity with category theory
(e.g.~\cite{AHS90}). We recall basic notions on distance functions and
coalgebras, and fix some notation.

\paragraph*{Distance functions and relations}
We write $\V$ to denote the unit interval
$[0,1] \subseteq \mathbb{R}$, with suprema and infima under the usual
ordering of the reals written as $\vee$ and $\wedge$, and $\oplus$ and
$\ominus$ used to denote truncated addition and subtraction on~$\V$,
respectively (i.e.~$x\oplus y=\min(x+y,1)$, $x\ominus y=\max(x-y,0)$).
We generally work with predicates $A \in 2^X$ where
$2= \{\top, \bot\}$, which we silently identify with subsets
$A \subseteq X$ when convenient.  Let $X,Y$ be sets. For relations
$R\subseteq X \times Y$, we write $R\colon \tworel{X}{Y}$, and we
write $R^\circ \colon \tworel{Y}{X}$ for the converse of a relation
$R \colon \tworel{X}{Y}$.  We then denote the relational composite of
two relations $R\colon \tworel{X}{Y}, S\colon \tworel{Y}{Z}$ by
$S \cdot R \colon \tworel{X}{Z}$. The relational image $R[A] \in 2^Y$
of $A\in 2^X$ under~$R$ is given by
$R[A]=\{y\mid\exists x\in A.\,x\mathrel{R}y\}$.

  A \emph{$\V$-valued relation} $r\colon X\times Y \to \V$ will be denoted by $r\colon \frel{X}{Y}$ and its \emph{converse} by $r^\circ \colon \frel{Y}{X}$, with $r^\circ(y,x) = r(x,y)$ for $x \in X$, $y \in Y$.
  A $\V$-valued relation $d \colon \frel{X}{X}$ is \emph{symmetric} if it coincides with its converse, and
  a \emph{hemimetric} if $d(x,x) = 0$ for all $x \in X$ and 
  $d(x,z) \leq d(x,y) \oplus d(y,z)$ for all $x,y,z \in X$;
  a \emph{pseudometric} is a symmetric hemimetric.
  Similarly to the qualitative case, the composite $s\cdot r \colon \frel{X}{Z}$ of $\V$-valued relations $r \colon \frel{X}{Y}, s\colon \frel{Y}{Z}$ is given by
  \begin{equation*}
    (s\cdot r)(x,z) = \bigwedge_{y\in Y} r(x,y) \oplus s(y,z).
  \end{equation*}

\paragraph*{Universal coalgebra}
  
Given an endofunctor $F\colon \set \to \set$, an \emph{$F$-coalgebra}
is a pair $(X, \xi)$, where $X$ is a set of \emph{states} and
$\xi \colon X \to FX$ is the \emph{transition} map. Coalgebras
generalize the notion of state-based systems: The transition map
defines how individual states can transition between each other, with
the transition dynamic (nondeterministic, probabilistic, etc.) being
specified by the endofunctor~$F$. We say that $(X,\xi)$ is
\emph{finitely branching} if for each $x\in X$, there exists a finite
subset $X_0\subseteq X$ and an element $a\in FX_0$ such that
$Fi(a)=\xi(x)$ where~$i$ denotes the inclusion map
$X_0\hookrightarrow X$. We may identify $FX_0$ with a subset
of~$FX$~\cite{Barr93}, and thus just write $\xi(x)\in FX_0$ in this
case. A functor~$F$ is \emph{finitary} if for every $a\in FX$, there
exists $FX_0$ such that $a\in FX_0$ in this sense. A homomorphism
between two $F$-coalgebras $(X, \xi),(Y, \delta)$ is a function
$h\colon X \to Y$ such that $\delta \circ h = Fh \circ \xi$.  We say
that two states $x\in X, y\in Y$ in coalgebras $(X, \xi), (Y, \delta)$
are behaviourally equivalent if there is a third coalgebra $(Z, \xi)$
and coalgebra homomorphisms $g\colon (X, \xi) \to (Z, \xi)$ and
$h\colon (Y,\delta) \to (Z, \xi)$ such that $g(x)=h(y)$.


\begin{expl}\label{expl:coalg}
  \begin{enumerate}[wide]
    \item The \emph{subdistribution functor} $\sfun\colon \set \to \set$ sends a set $X$  to the set of finitely supported subdistributions $\sfun X = \{\mu \colon X \to \V \mid \mu(x) \not = 0 \text{ for finitely many }x\in X, \sum_{x\in X} \mu(x) \leq 1\}$.
    It acts on functions $f\colon X \to Y$ by forming sums of preimages:
    $\sfun f (\mu)(y) = \sum_{x\in f^{-1}(y)} \mu(x)$. Coalgebras for $\sfun$ are Markov chains that optionally terminate with some probability in each state.
  \item\label{item:lts} Let $\pfun$ denote the (covariant) powerset
    functor, which maps a set~$X$ to its powerset $\pfun X$ and a map
    $f\colon X\to Y$ to the direct image map
    $\pfun f\colon \pfun X\to\pfun Y$ that takes direct images
    ($\pfun f(A)=f[A]$). Moreover, let~$\Act$ be a set of
    \emph{labels}. Then $\pfun(\Act\times {-})$-coalgebras are
    $\Act$-labelled transition systems. We will be particularly
    interested in the case where~$\Act$ is equipped with a metric, in
    which case $\Act$-labelled transition systems are often termed
    \emph{metric transition
      systems}~\cite{afs:linear-branching-metrics,FahrenbergLegay14}.
  \item Let $\dfun$ denote the distribution functor, which is defined
    analogously to the subdistribution functor~$\sfun$, but with the
    requirement that $\sum_{x\in X} \mu(x) = 1$. Coalgebras for the
    functor $\mathcal{D}(\Act\times {-})$ are (generative)
    probabilistic transition systems.
  \item\label{item:fuzzy} The \emph{fuzzy powerset functor} $\pfun_\V$
    sends a set $X$ to the set of functions $A\colon X\to \V$. On a
    function $f\colon X\to Y$, the functor is defined by
    $\pfun_\V f (A)(y) = \bigvee_{x\in f^{-1}(y)}A(x)$. Coalgebras for
    $\pfun_\V$ are fuzzy transition systems, i.e. transition systems
    in which adjacency takes on values in the unit interval.
  \end{enumerate}  
\end{expl}

For our discussions of behavioural distance and coalgebraic modal
logic, the notion of predicate lifting will play a crucial role.  In
the sequel, we fix a functor $F \colon \set \to \set$, and we denote
by $F^\Op \colon \set^\Op \to \set^\Op$ its opposite functor. (Recall
that $\set^\Op$ is the dual category of~$\set$, which has maps
$Y\to X$ as morphisms from~$X$ to~$Y$; then,~$F^\Op$ acts exactly
like~$F$, from which it differs only in its assigned type.) We write
$2^{-}$ for the \emph{contravariant powerset functor}
$\set^\Op\to\set$, which maps a set~$X$ to its powerset, identified
with the set $2^x$ of $2$-valued predicates on~$X$, and a function
$f\colon X\to Y$ to the function $2^f\colon 2^Y\to 2^X$ that takes
preimages ($2^f(B)=f^{-1}[B]$). Similarly, the \emph{contravariant
  $\V$-valued powerset functor} $\V^-$ maps a set~$X$ to the
set~$\V^X$ of $\V$-valued predicates on~$X$, and a map
$f\colon X\to Y$ to the preimage map~$\V^f\colon\V^Y\to\V^X$, given by
$\V^f(g)=g\cdot f$ for $g\colon Y\to\V$.

\begin{defn}
  A $2$-to-$2$ \emph{predicate lifting} for $F$ is a natural transformation of type $\lambda \colon 2^{-} \Rightarrow 2^{F^\Op -}$. A $\V$-to-$\V$ \emph{predicate lifting} is a natural transformation of type $\lambda \colon \V^{-} \Rightarrow \V^{F^\Op -}$.
  We say that a $2$-to-$2$, respectively $\V$-to-$\V$, predicate lifting is monotone if each of its components is monotone w.r.t.\ the pointwise orders induced by $2$, respectively by $\V$.
\end{defn}

\noindent The naturality condition on predicate liftings amounts to
commutation with preimage; e.g.~naturality of a $\V$-to-$\V$ predicate
lifting~$\lambda$ means that $\lambda(g\cdot f)=\lambda(g)\cdot Ff$
for all $g\in\V^Y$ and $f\colon X\to Y$. In (quantitative) coalgebraic
modal logic, the semantics of modalities is defined by attaching a
(quantitative) predicate lifting to every modal operator. Then
formulae of the form $\lambda\phi$ are interpreted in an $F$-coalgebra
$(X, \xi)$ via mappings $\Sem{\phi}_\xi\colon X \to \Omega$ where
$\Omega=2$ or $\Omega = \V$, defined inductively, with the case for
modal operators given by
$\Sem{\lambda\phi}_\xi(x) = \lambda(\Sem\phi_\xi)\circ \xi$.

 In \Cref{sec:kant}, we will use $\V$-to-$\V$ predicate liftings to construct quantitative lax extensions
which, as explained next, induce notions of behavioural distance for coalgebras.

\begin{defn}
  Let $L$ be a mapping of $\V$-valued relations $r\colon \frel{X}{Y}$ to $\V$-valued relations $Lr\colon \frel{FX}{FY}$. Then $L$ is called a  quantitative \emph{relator} of $F$ if it is monotone w.r.t.\ the pointwise order on $\V$-valued relations, i.e., if $r \leq r'$ entails $Lr \leq Lr'$, for all $\V$-valued relations $r,r' \colon \frel{X}{Y}$.
  Furthermore, $L$ is a quantitative \emph{lax extension} of $F$ if it is a quantitative relator that satisfies the following axioms for all $\V$-valued relations $r \colon \frel{X}{Y}$, $s \colon \frel{Y}{Z}$, and every function $f \colon X \to Y$,
  where a function $f$ is treated as a $\V$-valued relation with $f(x,y) = 0$ if $f(x)=y$ and $f(x,y) = 1$ otherwise:
    \begin{enumerate}
    \item $L(s \cdot r) \leq Ls \cdot Lr$,
    \item $Lf \leq Ff$ and $L(f^\circ) \leq (Ff)^\circ$.
    \end{enumerate}
\end{defn}

\noindent Quantitative relators yield notions of quantitative simulations for $F$-coalgebras:
Let $L$ be a quantitative relator, and let $(X,\alpha)$ and $(Y,\beta)$ be $F$-coalgebras.
A $\V$-valued relation $s \colon \frel{X}{Y}$ is an \emph{$L$-simulation} if $s \leq \beta^\circ \cdot L s \cdot f$, i.e.~if $s(x,y) \leq Ls(\alpha(x),\beta(y))$ for all $x \in X$ and $y \in Y$.
As $L$ is monotone, by the Knaster-Tarski fixed point theorem it follows that there is a smallest $L$-simulation between $(X,\alpha)$ and $(Y,\beta)$, wich we call \emph{$L$-behavioural distance} (from $(X,\alpha)$ to $(Y,\beta)$), and often denote by $d_L$.

\section{Threshold-Based Behavioural Distances}\label{sec:dist}

We proceed to introduce the central notion of \emph{threshold-based
  behavioural distances}. These distances generalize a distance on
(labelled) Markov chains that has been termed
\emph{$\epsilon$-distance}~\cite{DesharnaisEA08,DesharnaisSokolova25};
we will see that our general notion in fact subsumes many other
well-known behavioural distances on various system types
(\Cref{expl:L-Lambda}). As indicated above, they are based on first
introducing a notion of $\epsilon$-\mbox{(bi-)}simulation in which
probabilities under successor distributions are allowed to deviate by
at most~$\epsilon$, and then defining the induced (symmetric or
asymmetric) distance between states~$x,y$ as the infimum over
all~$\epsilon$ such that $x,y$ are $\epsilon$-\mbox{(bi-)}similar.

We parametrize the overall framework over a set functor~$F$
encapsulating the system type in a coalgebraic modelling as recalled
in \Cref{sec:preliminaries} and additionally over a choice of a
particular form of monotone modalities:
\begin{defn}
  A \emph{$2$-to-$\V$ predicate lifting} for the functor~$F$ is a
  natural transformation of type $2^{-}\to \V^{F^\Op}$, i.e.~a family~$\lambda$
  of maps
  \begin{equation*}
    \lambda_X\colon 2^X\to\V^{FX}
  \end{equation*}
  (recall that $\V=[0,1]$) indexed over all sets~$X$, subject to the
  \emph{naturality} condition that
  $\lambda_Y(A)(Ff(a))=\lambda_X(f^{-1}[A])(a)$ for $a\in FX$,
  $f\colon X\to Y$, and $A\in 2^Y$. We say that~$\lambda$ is
  \emph{monotone} if whenever $A\subseteq B$ for $A,B\in 2^X$, then
  $\lambda_X(A)(a)\le \lambda_X(B)(a)$ for all $a\in FX$. The
  \emph{dual}~$\overline\lambda$ of a $2$-to-$\V$ predicate
  lifting~$\lambda$ is given by
  \begin{equation*}
    \overline\lambda_X(A)(a)=1-\lambda(X\setminus A)(a)
  \end{equation*}
  for $A\in 2^X$, $a\in FX$. Given a set~$\Lambda$ of $2$-to-$\V$
  predicate liftings, we write~$\overline\Lambda$ for the closure
  of~$\Lambda$ under duals,
  i.e.~$\overline\Lambda=\Lambda\cup\{\overline\lambda\mid\lambda\in\Lambda\}$.
\end{defn}
\noindent \emph{For the remainder of the paper, we fix a functor~$F$
  and a set~$\Lambda$ of monotone $2$-to-$\V$ predicate liftings}. The
predicate liftings will serve as modalities in modal logics to be
introduced in \Cref{sec:two-valued-logic,sec:quant-log}; we will
generally use the terms \emph{predicate lifting} and \emph{modality}
interchangeably.  We restrict to unary liftings purely in the interest
of readability; the treatment of higher arities, i.e.~predicate
liftings of type $(2^{(-)})^n\to 2^{F^\Op}$, requires no more than
additional indexing. One basic example of a monotone $2$-to-$\V$
predicate lifting is \emph{probability}: Recall from
\Cref{sec:preliminaries} that coalgebras for the subdistribution
functor~$\sfun$ are Markov chains. We have a $2$-to-$\V$ predicate
lifting~$P$ for~$\sfun$ given by
\begin{equation*}
  P(A)(\mu)=\mu(A)\qquad\text{for $A\in 2^X$, $\mu\in\sfun X$}.
\end{equation*}
From the given choice of $2$-to-$\V$ predicate liftings, we obtain a
notion of (asymmetric) behavioural distance as follows:

\begin{defn}[$\epsilon$-(Bi-)simulation, threshold-based behavioural
  distance]\label{def:eps-sim}
  Let $\xi\colon X\to FX$, $\zeta\colon Y\to FY$ be $F$-coalgebras,
  and let $\epsilon\in\V$. A (two-valued) relation
  $R\colon\tworel{X}{Y}$ is an \emph{$\epsilon$-$\Lambda$-simulation}
  (from $(X,\xi)$ to $(Y,\zeta)$) if whenever $x\mathrel{R}y$, then
  \begin{equation*}
    \lambda(R[A])(\zeta(y))\ge\lambda(A)(\xi(x))-\epsilon\quad\text{for all $A\in 2^X$, $\lambda\in\Lambda$}.
  \end{equation*}
  A relation~$R$ as above is an
  \emph{$\epsilon$-$\Lambda$-bisimulation} if both~$R$ and its
  converse~$\rev{R}$ are $\epsilon$-$\Lambda$-simulations. If states
  $x\in X$, $y\in Y$ are related by some
  $\epsilon$-$\Lambda$-(bi-)simulation, then we say that~$y$
  \emph{$\epsilon$-$\Lambda$-simulates}~$x$ or that~$x$ and~$y$ are
  \emph{$\epsilon$-$\Lambda$-bisimilar}, and write
  $x\preceq_{\epsilon,\Lambda} y$ or $x\approx_{\epsilon,\Lambda}y$,
  respectively, defining $\preceq_{\epsilon,\Lambda}$ and
  $\approx_{\epsilon,\Lambda}$ as relations between~$X$ and~$Y$. We
  define the induced \emph{(threshold-based) $\Lambda$-behavioural
    distance} $\disteps{\Lambda}$ between states $x\in X$, $y\in Y$ by
  \begin{equation*}
    \disteps{\Lambda}(x,y)=\bigwedge\{\epsilon\mid x\preceq_{\epsilon,\Lambda} y\}.
  \end{equation*}
\end{defn}
\begin{rem}\label{rem:duals}
  Like in the two-valued case~\cite{GorinSchrode13} and in the purely
  quantitative setting~\cite{WildSchroder22}, it is easy to see that
  if~$\Lambda$ is closed under duals
  (i.e.~$\Lambda=\overline\Lambda$), then every
  $\epsilon$-$\Lambda$-simulation~$R$ is in fact an
  $\epsilon$-$\Lambda$-bisimulation. It follows that in this case, the
  threshold-based distance~$\disteps{\Lambda}$ is symmetric, hence a
  pseudometric. We thus focus almost entirely on
  $\epsilon$-$\Lambda$-simulations in the main line of the technical
  development, as we obtain corresponding results on bisimulations by
  just closing under duals.
\end{rem}
\begin{expl}\label{expl:epsilon-distance}
  The notion of $\epsilon$-$\Lambda$-simulation generalizes
  Desharnais et al.'s $\epsilon$-simulations on labelled Markov
  chains~\cite{DesharnaisEA08}. Specifically, a labelled Markov chain
  is a coalgebra for the functor~$F=\sfun^\Act$ \lsnote{$\sfun$ in
    prelims} where~$\Act$ is a set of labels and~$\sfun$ is the
  subdistribution functor (\Cref{sec:preliminaries}). Indeed, for each
  state~$c$ in a coalgebra $\xi\colon X\to FX$ and each action
  $\alpha\in\Act$, one has a subdistribution $\xi(c)(\alpha)$ over
  possible $\alpha$-successors of~$c$. In extension of the basic
  example of the $2$-to-$\V$ predicate lifting~$P$ for~$\sfun$, we
  have a set $\Lambda=\{P_\alpha\mid \alpha\in\Act\}$ of $2$-to-$\V$
  predicate liftings $P_{\alpha}$ for~$F$ given by
  \begin{equation*}
    P_\alpha(A)(f)=f(\alpha)(A)\qquad \text{for $A\in 2^X$ and $f\in FX=(\sfun X)^\Act$}.
  \end{equation*}
  Then, $\epsilon$-$\Lambda$-simulations on labelled Markov chains are
  precisely $\epsilon$-simulations in the sense of Desharnais et
  al. (see \Cref{rem:duals} for notions of
  $\epsilon$-bisimilarity). Explicitly, a relation
  $R\colon\tworel{X}{Y}$ is an $\epsilon$-$\Lambda$-simulation between
  labelled Markov chains $\xi\colon X\to FX$, $\zeta\colon Y\to FY$ is
  an $\epsilon$-$\Lambda$-simulation if whenever $x\mathrel{R} y$,
  then
  \begin{equation*}
    \xi(y)(a)(A)\ge\zeta(x)(a)(R[A])-\epsilon
  \end{equation*}
  for all $a\in\Act$, $A\in 2^X$.  For simplicity, we usually restrict
  to $|\Act|=1$, thus returning to the basic example of (unlabelled)
  Markov chains and~$\Lambda=\{P\}$ discussed earlier, a case that we
  continue to use as a running example; we defer the presentation of
  further instances to \Cref{expl:L-Lambda}. The dual~$\overline P$
  of~$P$ is given by $\overline P(A)(\mu)=1-\mu(X\setminus A)$; by
  \Cref{rem:duals}, $\epsilon$-$\overline\Lambda$-simulations are
  precisely $\epsilon$-bisimulations in the sense of Desharnais et al.
\end{expl}

\noindent We note some basic properties of
$\epsilon$-$\Lambda$-similarity:
\begin{lem}\label{lem:eps-sim}
  Let $\xi\colon X\to FX$, $\zeta\colon Y\to FY$ be
  $F$-coalgebras.
  \begin{enumerate}
  \item\label{item:gfp} The relation $\preceq_{\epsilon,\Lambda}$ is
    the greatest $\epsilon$-$\Lambda$-simulation from~$(X,\xi)$
    to~$(Y,\zeta)$.
  \item\label{item:eps-hemi} Threshold-quantitative
    $\Lambda$-behavioural distance is a (class-sized) hemimetric.
  \item\label{item:eps-monotone} Let $\delta\ge\epsilon\ge 0$, $x\in X$,
    $y\in Y$. If $x\preceq_{\epsilon,\Lambda}y$, then
    $x\preceq_{\delta,\Lambda}y$.
  \item\label{item:dist-sim} Let $x\in X$, $y\in Y$ such that
    $\disteps{\Lambda}(x,y)<\epsilon$. Then
    $x\preceq_{\epsilon,\Lambda}y$.
  \item\label{item:fb-sim} Suppose that~$(Y,\zeta)$ is finitely branching, and let
    $\epsilon\ge 0$, $x\in X$, $y\in X$. If
    $x\preceq_{\epsilon',\Lambda} y$ for every~$\epsilon'>\epsilon$,
    then $x\preceq_{\epsilon,\Lambda} y$. In particular,
    $\disteps{\Lambda}(x,y)\le\epsilon$ iff
    $x\preceq_{\epsilon,\Lambda}y$.
  \end{enumerate}
\end{lem}
\begin{proof}
  \begin{enumerate}[wide]
  \item By the Knaster-Tarski fixpoint theorem.
  \item The identity relation is an $\epsilon$-simulation for
    every~$\epsilon$, and the composite of an
    $\epsilon$-$\Lambda$-simulation and a
    $\delta$-$\Lambda$-simulation is an $\epsilon+\delta$-simulation.
  \item The relation $\preceq_{\epsilon,\Lambda}$ is a
    $\delta$-$\Lambda$-simulation.
  \item 
    By hypothesis, $x\preceq_{\epsilon',\Lambda}$ for some
    $\epsilon'<\epsilon$, and hence $x\preceq_{\epsilon,\Lambda}$ by
    Claim~\ref{item:dist-sim}.
  \item The second part of the claim is immediate from the first by
    Claim~\ref{item:dist-sim} and the definition
    of~$\disteps{\Lambda}$. For the first part of the claim, let $R$
    be the intersection of the relations $\preceq_{\epsilon',\Lambda}$
    for all $\epsilon'>\epsilon$. We show that~$R$ is an
    $\epsilon$-$\Lambda$-simulation. So let $A\in 2^X$,
    $\lambda\in\Lambda$; we have to show that
    $\lambda(R[A])(\zeta(y))\ge\lambda(A)(\xi(x))-\epsilon$.  Since
    $(Y,\zeta)$ is finitely branching, there exists a finite subset
    $Y_0\subseteq Y$ such that $\zeta(y)\in FY_0$, and by naturality
    of~$\lambda$, $\lambda(C)(\zeta(y))=\lambda( Y_0\cap C)(\zeta(y))$
    for every $C\in 2^Y$. By Claim~\ref{item:eps-monotone}, the finite
    sets $Y_0\cap{\preceq_{\epsilon',\Lambda}}[A]$ shrink as
    $\epsilon'>\epsilon$ tends towards~$\epsilon$, and thus become
    stationary at some point, say at $\epsilon_0>\epsilon$. Thus,
    $Y_0\cap R[A]=Y_0\cap{\preceq_{\epsilon',\Lambda}}[A]$ for
    every~$\epsilon'\in[\epsilon,\epsilon_0)$. Hence, by
    Claim~\ref{item:gfp}, we have
    \begin{align*}
      \lambda(R[A])(\zeta(y)) & =\lambda(Y_0\cap R[A])(\zeta(y))\\
                              & =\lambda(Y_0\cap{\preceq_{\epsilon',\Lambda}}[A])(\zeta(y))\\
                              & =\lambda({\preceq_{\epsilon',\Lambda}}[A])(\zeta(y)) \\
                              & \ge \lambda(A)(\xi(x))-\epsilon'
    \end{align*}
    for every $\epsilon'\in[\epsilon,\epsilon_0)$, which implies the
    claim. \qedhere
  \end{enumerate}
\end{proof}
\begin{rem}\label{rem:non-finitary}
  Without finite branching, \Cref{lem:eps-sim}.\ref{item:fb-sim} does not hold in general.
  Consider the functor $FX = \pfun([0,1]\times X)$ and the predicate lifting $\lambda_X(A)(B) = \sup \{ t \mid (t,x) \in B, x\in A \}$, and put $\Lambda=\{\lambda\}$.
  We define $F$-coalgebras $\xi\colon X\to FX$ and $\zeta\colon Y\to FY$,
  where $X = \{x\}$, $\xi(x) = \{(1,x)\}$,
  $Y = \{y\} \cup \{y_n \mid n \ge 1\}$,
  $\zeta(y) = \{(1,y_n) \mid n \ge 1\}$ and $\zeta(y_n) = \{(1-\frac{1}{n},y_n)\}$ for every $n \ge 1$.

  As $X$ is a singleton, we only need to consider the predicate $A = \{x\}$ to check claims of the form $x \preceq_{\epsilon,\Lambda} y'$.
  In particular, we have $x \preceq_{\epsilon,\Lambda} y_n \iff \frac{1}{n} \le \epsilon$.
  This implies that $x {\,\,\not\!\!{\preceq}}_{0,\Lambda} y$, as $\lambda({\preceq_{0,\Lambda}}[A])(\zeta(y)) = \lambda(\emptyset)(\zeta(y)) = 0 \not\ge 1 = \lambda(A)(\xi(x)) - 0$.
  We do however have ${\preceq_{\epsilon,\Lambda}}[A]\neq\emptyset$ for every $\epsilon > 0$, so that $x \preceq_{\epsilon,\Lambda} y$, and hence $d_\Lambda(x,y) = 0$.
\end{rem}

We capture the notion of $\epsilon$-$\Lambda$-similarity more
concisely using (two-valued) \emph{relators} over~$F$, which we just
understand as assignments~$L$ mapping relations $R\colon\tworel{X}{Y}$
to relations $LR\colon\tworel{FX}{FY}$, subject only to monotonicity
($R\subseteq S$ implies $LR\subseteq LS$), following usage in earlier
work~\cite{Thijs96,NoraEA25} (and noting that the term is sometimes
understood more
strictly~\cite{BackhouseBruinEtAl91,Levy11}). 
In this terminology, we have relators $\laxtwo_{\epsilon,\Lambda}$
over~$F$ for all $\epsilon\ge 0$, given by
\begin{equation*}
  a\mathrel{\laxtwo_{\epsilon,\Lambda} R}b\iff \forall \lambda\in\Lambda,A\subseteq X.\,\lambda(R[A])(b)\ge\lambda(A)(a)-\epsilon
\end{equation*}
for $R\subseteq X\times Y$, $a\in FX$, $b\in FY$. Then given
$F$-coalgebras $\xi\colon X\to FX$, $\zeta\colon Y\to FY$, a relation
$R\colon \tworel{X}{Y}$ is an $\epsilon$-$\Lambda$-simulation iff~$R$
is an \emph{$\laxtwo_{\epsilon,\Lambda}$-simulation},
i.e.~$\xi(x)\mathrel{\laxtwo_{\epsilon,\Lambda} R}\zeta(y)$ for all
$(x,y)\in R$.

We then define a quantitative lax extension $\laxv_\Lambda$ of~$F$
(\Cref{sec:preliminaries}) by
\begin{equation}\label{eq:laxv}
  \laxv_\Lambda r(a,b)  =\bigwedge\{\epsilon\in\V\mid a\mathrel{\laxtwo_{\epsilon,\Lambda} r_\epsilon} b\}
\end{equation}
for $r\colon\frel{X}{Y}$, $a\in FX$, $b\in FY$, where
\begin{equation*}
  r_\epsilon=\{(x,y)\in X\times Y\mid
  r(x,y)\le\epsilon\}
\end{equation*}
We defer the proof that $\laxv_\Lambda$ is in fact a lax extension --
we show later that~$\laxv_\Lambda$ is induced by a set of quantitative
modalities via the so-called Kantorovich construction
(\Cref{thm:lp-kantorovich-coalg}), which by general results always
yields a lax extension~\cite{WildSchroder22}.

We frequently consider cases where~$\Lambda$ consists only of a single
predicate lifting~$\lambda$; in this case, we write
$L_{\epsilon,\lambda}$ and~$L_\lambda$ in place of
$L_{\epsilon,\{\lambda\}}$ and $L_{\{\lambda\}}$, respectively. We
note that given $a\in FX$, $b\in FY$, and $r\colon\frel{X}{Y}$, the
set
$\{\epsilon\mid a\mathrel{\laxtwo_{\epsilon,\Lambda} r_\epsilon} b\}$
is upwards closed because both~$\laxtwo_{\epsilon,\Lambda}$ and
$r_\epsilon$ depend monotonously on~$\epsilon$.

\begin{rem}\label{rem:2-to-2}
  Every $2$-to-$2$ predicate lifting $\lambda\colon 2^{-} \to 2^{F^\Op-}$
  induces a $2$-to-$\V$ predicate lifting by postcomposition with the
  natural transformation $2^{F^\Op-}\into\V^{F^\Op-}$ defined by postcomposition
  with the inclusion of $2$ into $\V$; we say that a $2$-to-$\V$ predicate
  lifting~$\lambda$ is \emph{two-valued} if it is induced in this way,
  and we then also use $\lambda$ as a $2$-to-$2$ predicate lifting. If
  all predicate liftings in~$\Lambda$ are two-valued, then the
  relators $\laxtwo_{\epsilon,\Lambda}$ all coincide for $\epsilon<1$
  (indeed, they coincide with the two-valued lax extension induced by
  the corresponding $2$-to-$2$ predicate
  liftings~\cite{MartiVenema15}). Nevertheless,~$\laxv_\Lambda$ is
  typically properly quantitative, and in fact often coincides with
  familiar quantitative lax extensions. For instance, in case
  $F=\pfun$ and $\Lambda=\{\Diamond\}$
  where~$\Diamond\colon 2^{-}\to 2^{\pfun^\Op -}$ is the standard $2$-valued
  diamond modality given by
  $\Diamond_X(A)=\{S\in\pfun X\mid A\cap S\neq\emptyset\}$ for
  $A\in 2^X$, $\laxv_\Lambda$ is precisely the one-sided Hausdorff lax
  extension given by
  $\laxv_\Lambda r(S,T)=\bigvee_{x\in S}\bigwedge_{y\in T}r(x,y)$ for
  $S\in\pfun X$, $T\in\pfun Y$, and $r\colon\frel{X}{Y}$. By
  \Cref{rem:duals}, the closure~$\overline\Lambda$ of~$\Lambda$ under
  duals, $\overline\Lambda=\{\Diamond,\Box\}$, induces the symmetric
  Hausdorff lax extension, which, when applied to a metric on~$X$,
  yields the Hausdorff pseudometric on~$\pfun X$, given by
  $L_{\overline\Lambda}r(S,T)=\laxv_\Lambda r(S,T)\vee
  \laxv_\Lambda\rev r(T,S)=\bigvee_{x\in S}\bigwedge_{y\in
    T}r(x,y)\vee\bigvee_{y\in T}\bigwedge_{x\in S}r(x,y)$ for data as
  above.
\end{rem}
\begin{expl}\label{expl:L-Lambda}
  We describe a number of basic examples of threshold-based distances
  by specifying the relevant sets of $2$-to-$\V$ predicate
  liftings. Except in the case of the L\'evy-Prokhorov lifting, we
  defer more explicit descriptions of the induced threshold-based
  distances to \Cref{sec:kant}, where we will show that
  $\laxv_\Lambda$ is equivalently induced by the Sugeno modalities
  for~$\Lambda$. Explicit descriptions for most examples will then
  follow by identification of the Sugeno modalities with known
  quantitative modalities.
  \begin{enumerate}[wide]
  \item\label{item:L-Lambda-lp} Continuing
    \Cref{expl:epsilon-distance}, i.e.~taking~$F$ to be the
    subdistribution functor~$\sfun$ and $\Lambda=\{P\}$, we describe
    $\laxv_P$ as a one-sided form of the L\'evy-Prokhorov lifting, i.e.\
    we have
    \begin{equation*}
      \laxv_Pr(\mu,\nu)=\bigwedge\{\epsilon\mid\forall A\in 2^X.\, \nu(r_\epsilon[A])\ge\mu(A)-\epsilon\}
    \end{equation*}
    for $\mu\in\dfun X$, $\nu\in\sfun Y$, and
    $r\colon\frel{X}{Y}$. Correspondingly,  $\laxv_{\overline\Lambda}$
    is the usual two-sided L\'evy-Prokhorov lifting, i.e.
    \begin{multline*}
      \laxv_{\overline\Lambda}r(\mu,\nu)=\bigwedge\{\epsilon\mid\forall A\in 2^X.\, \nu(r_\epsilon[A])\ge\mu(A)-\epsilon\land\\
      \forall B\in 2^Y.\mu((\rev{r})_\epsilon[B])\ge\nu(B)-\epsilon\}.
    \end{multline*}
    We note here that over probability distributions, i.e.\
    subdistributions with total weight~$1$,~$P$ and~$\overline P$
    coincide, so that already~$\laxv_P$ defines the two-sided
    L\'evy-Prokhorov lifting.
  \item\label{item:L-Lambda-pts} Various further examples arise in the probabilistic setting. For
    instance, coalgebras for the functor~$F$ given by
    $FX=\dfun(\Act\times X)$ are generative probabilistic transition
    systems. We have monotone predicate liftings
    $\Diamond_a\colon 2^X\to\V^{FX}$ given by
    \begin{equation*}
      \Diamond_a(A)(\mu)=\mu\{(a,x)\mid x\in A\}.
    \end{equation*}
    This induces obvious variants of the one-sided L\'evy-Prokhorov
    lifting. For instance, for $\Lambda=\{\Diamond_a\mid a\in\Act\}$,
    $\mu\in FX$, $\mu\in FY$, and $r\colon\frel{X}{Y}$, we have
    $\laxv_\Lambda r(\mu,\nu)=\bigwedge\{\epsilon\mid\forall
    a\in\Act,A\in 2^X.\,\nu(\{a\}\times
    r_\epsilon[A])\ge\mu(\{a\}\times A)-\epsilon\}$.
  \item\label{item:L-Lambda-fts} \emph{Fuzzy transition systems:}
    Recall from \Cref{expl:coalg}.\ref{item:fuzzy} that coalgebras for
    the fuzzy powerset functor $\pfun_\V$ are fuzzy transition
    systems. Here, we have a natural $2$-to-$\V$ predicate lifting
    \begin{equation*}
      \Diamond(A)(g)=\bigvee_{x\in A}g(x).
    \end{equation*}
  \item\label{item:L-Lambda-mts} \emph{Metric transition systems:}
    Recall from \Cref{expl:coalg}.\ref{item:lts} that given a metric
    space~$\Act$ of labels, coalgebras for the
    functor$F=\pfun(\Act\times X)$ are metric transition systems. In
    this case, we have $2$-to-$\V$ predicate liftings
    \begin{equation*}
      \Diamond_a(A)(S)=\bigvee_{x\in A,(b,x)\in S}1-d(a,b)
    \end{equation*}
    saying roughly `there is a transition with label close to~$a$ into
    a state in~$A$'. The use of $1-d(a,b)$ is owed to the usual
    discrepancy between~$0$ representing complete equality as a
    distance value, but~$1$ representing complete truth.
  \item\label{item:L-Lambda-convex}
    The \emph{convex powerset functor} is the subfunctor $\cfun$ of $\pfun\dfun$ that maps a set $X$ to the set $\cfun X$ of (non-empty) convex subsets of $\dfun X$.
    We consider the $2$-to-$\V$ modality
    \begin{equation*}
      \Diamond(A)(V)=\sup_{\mu\in V}\mu(A).
    \end{equation*}
  \end{enumerate}
\end{expl}
\noindent We show next that the relator\/ $\laxv_\Lambda$ captures
$\epsilon$-$\Lambda$-similarity:
\begin{thm}\label{prop:lp-epsilon-sim}
  Let $(X,\xi)$ and $(Y,\zeta)$ be $F$-coalgebras. For states
  $x\in X$, $y\in Y$, we have
  \begin{equation*}
    \disteps{\Lambda}(x,y)=d_{\laxv_\Lambda}(x,y).
  \end{equation*}
\end{thm}
\noindent(Recall here that $\disteps{\Lambda}$ is the distance induced
by the notion of $\epsilon$-$\Lambda$-similarity as per
\Cref{def:eps-sim}, while $d_{\laxv_\Lambda}$ is the distance induced
by the quantitative lax extension $\laxv_\Lambda$ as recalled in
\Cref{sec:preliminaries}.)
\begin{proof}
  `$\le$': Let $r\colon\frel{X}{Y}$ be an $\laxv_\Lambda$-simulation,
  put $\epsilon=r(x,y)$, and let $\epsilon'>\epsilon$. It suffices to
  show that $x\preceq_{\epsilon',\Lambda}y$. To this end, we show that
  $r_{\epsilon'}$ is an $\epsilon'$-$\Lambda$-simulation. So let
  $x'\mathrel{r_{\epsilon'}} y'$, $A\in 2^X$, $\lambda\in\Lambda$; we
  have to show that
  \begin{equation}
    \label{eq:goal-eps-prime}
    \lambda(r_{\epsilon'}[A])(\zeta(y'))\ge\lambda(A)(\xi(x'))-\epsilon'.
  \end{equation}
  Since $r$ is an $\laxv_\Lambda$-simulation, we have
  \begin{equation*}
    \laxv_\Lambda r(\xi(x'),\zeta(y'))\le r(x',y') <\epsilon'.
  \end{equation*}
  Thus, there exists $\epsilon''<\epsilon'$ such that
  $\xi(x')\mathrel{L_{\epsilon'',\Lambda}r_{\epsilon''}}\zeta(y')$. Hence,
  $\lambda(r_{\epsilon''}[A])(\zeta(y'))
  \ge\lambda(A)(\xi(x'))-\epsilon''$. Since
  $r_{\epsilon''}[A]\subseteq r_{\epsilon'}[A]$, this implies our
  goal~\eqref{eq:goal-eps-prime}.

  `$\ge$': Let $R\colon\tworel{X}{Y}$ be an
  $\epsilon$-$\Lambda$-simulation. Define
  $r\colon\frel{X}{Y}$ by
  \begin{equation*}
    r(x,y)=
    \begin{cases}
      \epsilon & x\mathrel{R}y\\
      1 & \text{otherwise.}
    \end{cases}
  \end{equation*}
  It suffices to show that $r$ is an
  $\laxv_\Lambda$-simulation, i.e.\
  $\laxv_\Lambda r(\xi(x),\zeta(y))\le
  r(x,y)$. So suppose that $x\mathrel{R} y$
  (otherwise, there is nothing to show); we have to show that
  $\laxv_\Lambda r(\xi(x),\zeta(y))\le
  \epsilon$. By definition of~$\laxv_\Lambda$, this follows once we
  show that
  $(\xi(x),\zeta(y))\in \laxtwo_{\epsilon,\Lambda}
  r_\epsilon$. But this is immediate from the
  fact that~$R$ is an $\epsilon$-$\Lambda$-simulation,
  since~\mbox{$r_\epsilon=R$}.
\end{proof}

\begin{expl}
  By \Cref{expl:L-Lambda}.\ref{item:L-Lambda-lp}, one instance of
  \Cref{prop:lp-epsilon-sim} yields the characterization of
  $\epsilon$-distance on labelled Markov chains~\cite{DesharnaisEA08}
  via fixpoints of the L\'evy-Prokhorov lifting established recently
  by Desharnais and Sokolova~\cite{DesharnaisSokolova25}; indeed we
  obtain this both for the two-sided (symmetric) and for the one-sided
  (asymmetric) variant, respectively applying
  \Cref{prop:lp-epsilon-sim} to $\{P,\overline P\}$ or to
  $\{P\}$. Correspondingly, we occasionally refer to this behavioural
  distance as \emph{L\'evy-Prokhorov behavioural distance}, adding the
  quantification \emph{one-sided} or \emph{two-sided} when needed.
\end{expl}


\section{Threshold-Quantitative Codensity Games}
\noindent With a view to using Spoiler strategies in the computation
of distinguishing formulae, we next introduce a Spoiler-Duplicator
game for $\epsilon$-similarity.
\begin{defn}\label{def:codensity-game}
  Let $\Lambda$ be a set of $2$-to-$\V$ predicate liftings, and let
  $\xi\colon X\to FX$, $\zeta\colon Y\to FY$ be $F$-coalgebras.  The
  \emph{codensity game up to~$\epsilon$} on~$\xi$ and~$\zeta$, played
  by Spoiler (S) and Duplicator (D), is given by the following table
  detailing positions, ownership\bknote{owner = player? LS: I wouldn't really know what the player of a position is; what I mean by `owner' is the player whose turn it is in the position. As far as I recall, this is standard}, and moves:
  \begin{center}
    \renewcommand{\arraystretch}{1.2}
    \scalebox{0.93}{
    \begin{tabular}{|l|c|l|}
      \hline
      Position & Owner & Moves\\\hline
      $(x,y)\in X\times Y$ & S & $\{(A,B)\in 2^X\times 2^Y\mid\exists\lambda
                                 \in\Lambda.$\\ && $\quad\lambda(B)(\zeta(y))<\lambda(A)(\xi(x))-\epsilon\}$\\\hline
      $(A,B)\in 2^X\times 2^Y$ & D
                       & $\{(x,y)\in X\times Y\mid x\in A, y\notin B\}$\\\hline
     \end{tabular}
     }
  \end{center}
  As usual, any player who cannot move on their turn loses. Infinite
  plays are won by~D.
\end{defn}
\noindent Intuitively, we understand a position of the form $(x,y)$
as~D claiming that~$x$ and~$y$ are $\epsilon$-$\Lambda$-similar,
i.e.~$x\preceq_{\epsilon,\Lambda}y$. Spoiler can challenge such a
claim by exhibiting predicates~$A$ and~$B$ such that
$\lambda(B)(\zeta(y))<\lambda(A)(\xi(x))-\epsilon$, which refutes D's
claim provided that
${\preceq_{\epsilon,\Lambda}^\epsilon}[A]\subseteq B$ (here,
${\preceq_{\epsilon,\Lambda}^\epsilon}[A]\subseteq Y$ is the set of
states that $\epsilon$-$\Lambda$-simulate some state in~$A$). Thus,~D
needs to challenge this inclusion by playing $(x',y')$ where $x'\in A$
and $y'\notin B$, thereby claiming that
$x'\preceq_{\epsilon,\Lambda} y'$, i.e.~that
$y'\in {\preceq_{\epsilon,\Lambda}}[A]$.

\begin{thm}
  Let $\xi\colon X\to FX$, $\zeta\colon Y\to FY$ be $F$-coalgebras.
  Duplicator wins the position $(x,y)$ in the codensity game up
  to~$\epsilon$ on~$\xi$ and~$\zeta$ iff $y$
  $\epsilon$-$\Lambda$-simulates~$x$.
\end{thm}
\begin{proof}
  \emph{`If':} We show that given an
  $\epsilon$-$\Lambda$-simulation~$R$ between~$X$ and~$Y$, Duplicator
  wins all positions in~$R$ by enforcing the invariant $(x,y)\in R$
  for S-positions $(x,y)$. To see that the invariant can be enforced,
  suppose that at $(x,y)\in R$,~S plays $(A,B)$; then we
  have~$\lambda\in\Lambda$ such that
  $\lambda(B)(\zeta(y))<\lambda(A)(\xi(x))-\epsilon$. Since~$R$ is an
  $\epsilon$-$\Lambda$-simulation, it follows that
  $R[A]\not\subseteq B$, so there exists $(x,y)\in R$ such that
  $x\in A$ and $y\notin B$; that is,~D can enforce the invariant by
  playing $(x,y)$. The same argument shows that the invariant
  guarantees that~D can always reply to moves of~S, so~D wins.

  \emph{`Only if':} We show that the set~$R$ of S-positions won by~D
  is an $\epsilon$-$\Lambda$-simulation. Indeed, let $(x,y)\in R$,
  $\lambda\in\Lambda$, and $A\in 2^X$. We have to show that
  $\lambda(R[A])(\zeta(y))\ge \lambda(A)(\xi(x))-\epsilon$. Assume
  the contrary. Then S can play $(A,R[A])$ at $(x,y)$. To this move,~D
  has no reply, contradicting the assumption that D wins $(x,y)$.
\end{proof}
\noindent In connection with \Cref{lem:eps-sim}.\ref{item:fb-sim}, we
obtain
\begin{cor}[Game correctness]\label{cor:game-correctness-finite}
  Let~$\Lambda$ be a set of\/ $2$-to-$\V$ predicate liftings, and let
  $\xi\colon X\to FX$, $\zeta\colon Y\to FY$ be finitely branching
  $F$-coalgebras.  Duplicator wins the position $(x,y)$ in the
  codensity game up to~$\epsilon$ on~$\xi$ and~$\zeta$ iff
  $\disteps{\Lambda}(x,y)\le\epsilon$.
\end{cor}
\begin{rem}
  We expect that the codensity game up to $\epsilon$ can be cast as an
  instance of the very general notion of codensity games on
  fibrations~\cite{KomoridaEA19}, in this case the fibration of binary
  relations over the base category $\set^2$ of pairs of sets (that is,
  the fibre over a pair $(X,Y)$ of sets consists of the relations of
  type $R\colon\tworel{X}{Y}$). Desharnais et
  al.~\cite{DesharnaisEA08} present a game for the specific case of
  L\'evy-Prokhorov distance, which however is quite different in
  spirit from the codensity game up to~$\epsilon$.
\end{rem}

We will later extract distinguishing formulae from winning strategies
for Spoiler. We are thus interested in the efficient computation of
Spoiler strategies. Since the codensity game is a safety game,
Spoilers winning region is a least fixed
point, and thus can be computed by Kleene iteration.
We have the following generic time estimate for this computation:
\begin{defn}\label{def:poly-solvable}
  A set~$\Lambda$ of $2$-to-$\V$ predicate liftings is
  \emph{polynomial-time solvable} if given $\epsilon\ge 0$, functor
  elements $a\in FX$, $b\in FY$, and a relation
  $S\colon\tworel{X}{Y}$, one can decide in polynomial time (in the
  total size of the input, made up of $a$,~$b$,~$S$, and~$\epsilon$,
  the latter represented in binary) whether there exist $A\in 2^X$,
  $B\in 2^Y$, $\lambda\in\Lambda$ such that
  $\lambda(B)(b)<\lambda(A)(a)-\epsilon$ and
  $A\times(Y\setminus B)\subseteq S$, and moreover in this case one
  can compute these data in polynomial time.
\end{defn}
\noindent In the examples below, we typically establish polynomial-time solvability by forming the complement relation $R = (X\times Y)\setminus S$, and taking $B$ to be $R[A]$.
This ensures the condition $A\times(Y\setminus B) \subseteq S$, so that only the inequality needs to be checked.
\begin{thm}[Strategy computation]\label{thm:strat}
  Let $\xi\colon X\to FX$, $\zeta\colon Y\to FY$ be finite coalgebras,
  and let $\epsilon\ge 0$. If~$\Lambda$ is polynomial-time solvable,
  then both the winning region of Spoiler in the codensity game up
  to~$\epsilon$ on~$\xi$ and~$\zeta$ and a history-free winning
  strategy on the winning region are computable in polynomial time.
\end{thm}
\begin{proof}
  As indicated above, the winning region of Spoiler is computed as a
  least fixpoint by iterating the functional that maps a relation~$S$
  to the set of all positions~$(x,y)$ that spoiler can force into~$S$
  in the next move by means of some move~$(A,B)$. This functional, as
  well as the move $(A,B)$, are computable in polynomial time by
  polynomial solvability of~$\Lambda$. In the Kleene iteration, the
  moves $(A,B)$ computed in each step form Spoiler's winning
  strategy. There are at most $|X\times Y|$ many iterations, so
  we have overall polynomial runtime.
\end{proof}
\noindent By the correctness of the game
(\Cref{cor:game-correctness-finite}), this entails:
\begin{cor}
  If~$\Lambda$ is polynomially solvable, then it is decidable in
  polynomial time whether two given states in finite coalgebras have
  distance at most~$\epsilon$.
\end{cor}

\begin{expl}\label{expl:poly-solvable}
  \begin{enumerate}[wide]
  \item \emph{L\'evy-Prokhorov behavioural distance} on Markov chains:
    $\Lambda=\{P\}$ is polynomially solvable. Indeed, given
    $\mu\in\sfun X$, $\nu\in\sfun Y$, and $R\subseteq X\times Y$,
    Desharnais et al.~\cite{DesharnaisEA08} construct a flow network
    $\mathcal N(\mu,\nu,R)$ whose maximal flow is at least
    $\mu(X)-\epsilon$ iff for all $A\in 2^X$,
    $P(R[A])(\nu)\ge P(A)(\mu)-\epsilon$,
    i.e.~$\nu(R[A])\ge\mu(A)-\epsilon$. Given the data from
    \Cref{def:poly-solvable}, we apply this
    to~$R=(X\times Y)\setminus S$. In case the maximal flow (which can
    be computed in polynomial time by standard methods) is less than
    $\mu(X)-\epsilon$, one obtains $A\in 2^X$, $B\in 2^Y$ such that
    $\nu(B)<\mu(A)-\epsilon$ and $A\times(Y\setminus B)\subseteq S$
    from a minimum cut of $\mathcal N(\mu,\nu,R)$ (which in turn is
    read off from a maximal flow in a standard manner); details are in
    the appendix. Modifications for the two-sided case
    ($\Lambda=\{P,\overline P\}$) are straightforward.
  \item The case of generative probabilistic transition systems, i.e.\ $FX = \dfun(\Act\times X)$, can be solved using the technique from the previous item.
    Explicitly, if $\mu\in\dfun(\Act\times X)$, $\nu\in\dfun(\Act\times Y)$ and $S\colon\tworel{X}{Y}$, then we can,
    for each $a\in\Act$, restrict $\mu$ and $\nu$ to the subsets $\{a\}\times X$ and $\{a\}\times Y$, respectively, and check if a pair $(A,B)$ satisfying the condition for the two resulting subdistributions exists.
    Note that we do not need to check labels $a$ that do not occur in either of the supports of $\mu$ and $\nu$.
  \item The modality $\Diamond(A)(f) = \bigvee_{x\in A} f(x)$
    for the fuzzy powerset functor $\pfun_\V$ is polynomially solvable.
    Let $g\in\pfun_\V X$, $h\in\pfun_\V Y$ and $S\colon\tworel{X}{Y}$ and let $\epsilon \ge 0$.
    The key observation is that we have $\Diamond(S[A])(h) \ge \Diamond(A)(g) - \epsilon$ for all $A\in 2^X$ iff we have this for all singleton sets.
    In the nontrivial direction, let $A\in 2^X$. Then we have
    \begin{multline*}
      \Diamond(A)(g) = \bigvee_{x\in A} g(x) - \epsilon \\
      = \bigvee_{x\in A} \bigvee_{y\in S[\{x\}]} h(y)
      = \bigvee_{y\in S[A]} h(y) = \Diamond(S[A])(h).
    \end{multline*}
  \item The set of modalities $\Diamond_a(A)(S) = \bigvee_{x\in A,(b,x)\in S} 1-d(a,b)$ for metric transition systems is polynomially solvable.
    This is because, like in the previous item, the modality is based on taking suprema and it therefore suffices to check singletons.
    Similar to before, we skip labels not occurring in the given functor elements.
  \end{enumerate}
\end{expl}

\section{Two-Valued Distinguishing
  Formulae}\label{sec:two-valued-logic}

We now introduce a two-valued coalgebraic modal logic equipped with a
notion of satisfaction up to~$\epsilon$, with a view to certifying
high behavioural distances by modal formulae (a quantitative logic for
similar purposes is discussed in \Cref{sec:quant-log}). To this end,
we (implicitly) convert the given $2$-to-$\V$ predicate liftings into
$2$-to-$2$ predicate lifting by means of threshold values; as a simple
example, for the probability modality $\Lambda=\{P\}$ over
subdistributions, one obtains the usual threshold modalities `with
probability at least~$q$' known from two-valued probabilistic modal
logics~\cite{LarsenSkou91}. Satisfaction up to~$\epsilon$ then means
essentially that~$\epsilon$ is subtracted from all thresholds.

Formally, we define the two-valued logic $\Lang_2(\Lambda)$ as
follows. The set $\FLang(\Lambda)$ of \emph{$\Lambda$-formulae}
$\phi,\psi,\dots$ is given by the grammar
\begin{equation*}
  \phi,\psi::=\bot\mid\top\mid\phi\land\psi\mid\phi\lor\psi\mid\lambda_q\phi\qquad(\lambda\in\Lambda,q\in\V).
\end{equation*}
The \emph{modal rank} of a formula~$\phi$ is the maximal nesting depth
of modalities~$\lambda_q$ in~$\phi$.

We interpret this logic in terms of satisfaction relations
up-to-$\epsilon$, denoted $\models_\epsilon$, between states~$x$ in
$F$-coalgebras $\xi\colon X\to FX$ and $\Lambda$-formulae. These
relations are recursively defined:
\begin{align*}
  x&\models_\epsilon\top\\
  x&\not\models_\epsilon\bot\\
  x&\models_\epsilon\phi\land\psi\quad\text{iff}\quad x\models_\epsilon\phi\text{ and }x\models_\epsilon\psi\\
  x&\models_\epsilon\phi\lor\psi\quad\text{iff}\quad x\models_\epsilon\phi\text{ or }x \models_\epsilon\psi\\
  x&\models_\epsilon\lambda_q\phi\quad\text{iff}\quad\lambda(\Sem{\phi}_\epsilon)(\xi(x))\ge q\ominus\epsilon
\end{align*}
where $\Sem{\phi}_\epsilon=\{z\in X\mid z\models_\epsilon\phi\}$. (Of
course, we could equivalently replace $q\ominus\epsilon$ with
$q-\epsilon$ in the last clause, but some arguments are simplified by
writing $q\ominus\epsilon$.)  Intuitively, $x \models_\epsilon\phi$
means that state $x$ might not satisfy $\phi$ exactly, but satisfies
it at least up to $\epsilon$. The smaller the $\epsilon$, the closer
we are to classical satisfaction ($\models_0$). Whenever
$\epsilon\le \delta$, $x \models_\epsilon\phi$ implies
$x\models_\delta \phi$; that is,
$\Sem{\phi}_\epsilon\subseteq \Sem{\phi}_\delta$.

\begin{defn}
  Given states $x\in X$, $y\in Y$ in $F$-coalgebras
  $\xi\colon X\to FX$, $\zeta\colon Y\to FY$, respectively, we say
  that~$x$ is \emph{logically included in~$y$ up to $\epsilon$} if for
  all $\phi\in\FLang(\Lambda),\delta\in\V$, whenever $x\models_\delta\phi$, then
  $y\models_{\delta+\epsilon}\phi$.
\end{defn}
\noindent We will show that in the above definition,~$\delta$ may
equivalently be fixed to be~$0$. This relies on the following
observation:

\begin{defn}
  For $\delta\in\V$, we define the \emph{$\delta$-relaxation}
  $r_\delta(\phi)$ of a $\Lambda$-formula~$\phi$ recursively by
  \begin{equation*}
    r_\delta(\lambda_q\phi)=\lambda_{q\ominus \delta}r_\delta(\phi)
  \end{equation*}
  and commutation with all other constructs.

\end{defn}
\begin{lem}\label{lem:relax}
  Let $\xi\colon X\to FX$ be an $F$-coalgebra, and let
  $\epsilon,\delta\in\V$ such that $\delta\le\epsilon$. Then
  \begin{equation*}
    x\models_\epsilon\phi\quad\text{iff}\quad x\models_{\epsilon-\delta}r_\delta(\phi)
  \end{equation*}
  for every $x\in X$ and $\phi\in\FLang(\Lambda)$.
\end{lem}
\begin{proof}
  Induction on~$\phi$. The Boolean cases are trivial. For the
  modalities~$\lambda_q$, we reason as follows. Let
  $x\models_\epsilon\lambda_q\phi$, that is,
  $\lambda(\Sem{\phi}_\epsilon)(\xi(x))\ge q-\epsilon$. We have to
  show that
  $x\models_{\epsilon-\delta}r_\delta(\lambda_q\phi)=\lambda_{q\ominus
    \delta}r_\delta(\phi)$, i.e.\ that
  $\lambda(\Sem{r_\delta(\phi)}_{\epsilon-\delta}(\xi(x))\ge (q\ominus
  \delta)-(\epsilon-\delta)$. By induction,
  $\Sem{\phi}_\epsilon=\Sem{r_\delta(\phi)}_{\epsilon-\delta}$. We are thus done
  once we show that
  \begin{equation*}
    (q\ominus \delta)\ominus(\epsilon-\delta)=q\ominus\epsilon.
  \end{equation*}
  We distinguish cases on how the various occurrences of~$\ominus$
  evaluate. If $\epsilon\le q$, then also $\delta\le q$ and
  $\epsilon-\delta\le q-\delta = q\ominus \delta$, so all~$\ominus$ are in fact~$-$,
  and the claim is clear. So suppose that $\epsilon>q$. Then the right
  hand side is~$0$. Moreover, since $\delta\le\epsilon$, we have
  $\epsilon-\delta=\epsilon\ominus \delta\ge q\ominus \delta$, so the left hand side
  is also~$0$.
\end{proof}

\begin{lem}\label{lem:logical-zero}
  Let $x,y$ be states in finitely branching $F$-coalgebras
  $\xi\colon X\to FX$, $\zeta\colon Y\to FY$, respectively, and let
  $\epsilon\in\V$. Then~$y$ logically $\epsilon$-simulates~$x$ iff for
  all $\phi\in\FLang(\Lambda)$, whenever $x\models_0\phi$, then
  $y\models_\epsilon\phi$.
\end{lem}
\begin{proof}
  `Only if' is trivial; we prove `if'. So let $\delta\in V$ and
  $\phi\in\FLang(\Lambda)$ such that $x\models_\delta\phi$. We have to
  show that
  $y\models_{\delta+\epsilon}\phi$. 
  By Lemma~\ref{lem:relax}, we obtain $x\models_0 r_\delta(\phi)$. By
  hypothesis, we obtain $y\models_\epsilon r_\delta(\phi)$, and again by
  Lemma~\ref{lem:relax}, we conclude $y\models_{\delta+\epsilon}\phi$, as
  required.
\end{proof}
\noindent Logical inclusion satisfies the expected continuity
property:
\begin{lem}\label{lem:sat-epsilon}
  Let $\xi\colon X\to FX$, $\zeta\colon Y\to FY$, let $x\in X$,
  $y\in Y$, and let $\epsilon\in\V$.
  \begin{enumerate}
  \item If $\epsilon\le\epsilon'$ and $x$ is logically included in~$y$
    up to~$\epsilon$, then~$x$ is logically included in~$y$ up
    to~$\epsilon'$.
  \item\label{item:epsilon-cont} If $x$ is logically included in~$y$
    up to~$\epsilon'$ for all $\epsilon'>\epsilon$, then~$x$ is
    logically included in~$y$ up to~$\epsilon$.
  \end{enumerate}
\end{lem}

\noindent Using this fact, we prove that satisfaction is preserved up
to behavioural distance:

\begin{lem}[Preservation]\label{lem:invariance}
  Let $x\in X$, $y\in Y$ be states in $F$-coalgebras
  $\xi\colon X\to FX$, $\zeta\colon Y\to FY$, and let
  $\epsilon\in\V$. If $\disteps{\Lambda}(x,y)\le\epsilon$, then~$x$ is
  logically included in~$y$ up to~$\epsilon$.
\end{lem}
\begin{proof}
  By \Cref{lem:sat-epsilon}.\ref{item:epsilon-cont}, it suffices to
  prove that whenever~$y$ $\epsilon$-$\Lambda$-simulates~$x$, then~$x$
  is logically included in~$y$ up to~$\epsilon$. So let~$R$ be an
  $\epsilon$-$\Lambda$-simulation such that $x\mathrel{R}y$.  We use
  Lemma~\ref{lem:logical-zero} and show by induction on~$\phi$ that
  whenever $x\models_0\phi$, then~$y\models_\epsilon\phi$. Boolean
  cases are trivial. For the modal case, suppose that
  $x\models_0\lambda_q\phi$, i.e.\
  $\lambda(\Sem{\phi}_0)(\xi(x))\ge q$. Since~$R$ is an
  $\epsilon$-$\Lambda$-simulation, this implies
  $\lambda(R[\Sem{\phi}_0])(\zeta(y))\ge q-\epsilon$. By induction, we
  have $R[\Sem{\phi}_0]\subseteq\Sem{\phi}_\epsilon$, so by
  monotonicity of~$\lambda$,
  $\lambda(\Sem{\phi}_\epsilon)(\zeta(y))\ge q-\epsilon$; that is,
  $y\models_\epsilon\lambda_q\phi$ as required.
\end{proof}
\noindent Indeed, we can show that over finitely branching coalgebras,
the threshold-based $\Lambda$-behavioural distance~$\disteps{\Lambda}$
is characterized by the logic~$\Lang(\Lambda)$, a property often
referred to as \emph{expressiveness} of~$\Lang(\Lambda)$:
\begin{defn}
  Let $x\in X$, $y\in Y$ be states in $F$-coalgebras
  $\xi\colon X\to FX$, $\zeta\colon Y\to FY$, respectively. A formula
  $\phi\in\FLang(\Lambda)$ is \emph{$\epsilon$-distinguishing} for
  $(x,y)$ if $x\models_0\phi$ and
  $y\not\models_\epsilon\phi$.
\end{defn}

\begin{thm}[Expressiveness, two-valued]\label{thm:hm-two-valued}
  Let $x\in X$, $y\in Y$ be states in finitely branching
  $F$-coalgebras $\xi\colon X\to FX$, $\zeta\colon Y\to
  FY$. Then~$\disteps{\Lambda}(x,y)\le\epsilon$ iff~$y$ logically
  $\epsilon$-simulates~$x$.
\end{thm}

\noindent That is, whenever $\disteps{\Lambda}(x,y)>\epsilon$, then
there exists an $\epsilon$-distinguishing formula for $(x,y)$ in
$\Lang(\Lambda)$. A proof (very similar to that of
\Cref{thm:extract-two-valued}) is in the appendix. For our present
purposes, we focus instead on the computation of distinguishing
formulae in case the involved coalgebras $\xi\colon X\to FX$,
$\zeta\colon Y\to FY$ are finite.  Let~$s$ be a memoryless winning
strategy for Spoiler that wins all positions in the winning region of
Spoiler in the codensity game up to~$\epsilon$ on~$\xi$ and~$\zeta$ as
per \cref{def:codensity-game} (such a strategy exists because from
Spoiler's perspective, the game is a reachability game). Since there
are only $k=|X|\times|Y|$ Spoiler positions in the game, Spoiler wins
any position in their winning region in at most~$k$ moves. We compute
$\epsilon$-distinguishing formulae for all pairs $(x,y)$ of states in
Spoiler's winning region in dag representation, i.e.~sharing identical
subformulae (cf.~\Cref{rem:dags}), using dynamic programming, with
stages indexed by the remaining number of Spoiler moves. As positions
of the shape $(x,y)\in X\times Y$ belong to Spoiler, this number is
always greater than~$0$. In stage $i>0$, the algorithm proceeds as
follows for each position $(x_0,y_0)$ won by~$S$ in~$i$ moves:

\begin{alg}[Extraction of two-valued distinguishing
  formulae]\label{alg:two-valued-extraction} ~

  \begin{enumerate}
  \item Put $(A,B)=s(x_0,y_0)$ (Spoiler's winning move). By definition
    of the game, there exists $\lambda\in\Lambda$ such that
    $\lambda(B)(\zeta(y))<q-\epsilon$ where $q=\lambda(A)(\xi(x))$.
  \item For every possible reply by~$D$, i.e.~every pair
    $(x,y)\in X\times Y$ such that $x\in A$ and $y\notin B$,~$S$ wins
    $(x,y)$ in at most $i-1$ moves, so we have already computed an
    $\epsilon$-distinguishing formula $\phi_{xy}$ for $(x,y)$ in
    previous stages. 
  \item Put
    \begin{equation}
      \label{eq:two-valued-distinguishing-phi}
      \textstyle \phi=\bigvee_{x\in A}\bigwedge_{y\in Y\setminus B}\,\phi_{xy}.
    \end{equation}
  \item Assign the $\epsilon$-distinguishing formula
    $\phi_{x_0y_0}=\lambda_q\phi$ to $(x_0,y_0)$.
  \end{enumerate}
\end{alg}
\noindent The correctness of the algorithm and the complexity of the
formulae it computes are formulated as follows:
\begin{thm}[Polynomial-time extraction of distinguishing
  formulae]\label{thm:extract-two-valued}
  Given finite $F$-coalgebras $\xi\colon X\to FX$,
  $\zeta\colon Y\to FY$ and a Spoiler strategy~$s$ for the codensity
  game up to~$\epsilon$ on~$\xi$ and~$\zeta$ that wins all positions
  in Spoiler's winning region, \cref{alg:two-valued-extraction}
  constructs two-valued $\epsilon$-distinguishing
  formulae~$\phi_{xy}\in\FLang(\Lambda)$ for all pairs of states
  $(x,y)\in X\times Y$ such that
  $\disteps{\Lambda}(x,y)>\epsilon$. The algorithm runs in polynomial
  time, and the formulae~$\phi_{x,y}$ it constructs are of quadratic
  modal rank and polynomial, specifically quartic,
  dag size.
\end{thm}
\begin{proof}
  The complexity estimates are immediate. We prove the correctness
  claim by induction on the stage~$i$, i.e.~we show that~$\phi$ as
  per~\eqref{eq:two-valued-distinguishing-phi} is an $\epsilon$-distinguishing
  formula for $(x_0,y_0)$. 
  By construction, we have $A\subseteq\Sem{\phi}_0$, and therefore
  $x\models_0\lambda_q\phi$. On the other hand, we claim that
  $y_0\not\models_\epsilon\lambda_q\phi$, that is,
  $\lambda(\Sem{\phi}_{q-\epsilon})(\zeta(y_0))<q-\epsilon$. Since
  $(A,B)$ is a legal move for Spoiler at $(x_0,y_0)$, we have
  $\lambda(B)(\zeta(y_0))<q-\epsilon$, so we are done once we show
  that $\Sem{\phi}_{q-\epsilon}\subseteq B$. We show the
  contrapositive: Let $y\in Y\setminus B$; we show that
  $y\not\models_{q-\epsilon}\phi$. So let $x\in A$; we have to show
  that
  $y\not\models_{q-\epsilon}\bigwedge_{y'\in Y\setminus B}\phi_{xy'}$,
  which however follows from $y\not\models_{q-\epsilon}\phi_{xy}$.
\end{proof}

\begin{rem}[Dag representation]\label{rem:dags}
  As indicated above, we measure formulae in \emph{dag size}, i.e.~the
  size of a representation where identical subformulae are
  shared. This measure is similar to just counting the number of
  subformulae except that it takes into account the representation
  size of modalities. In terms of implementation, this will usually
  mean that formulae are represented as data structures on the
  heap. The size of the formula without sharing of subformulae is
  often referred to as \emph{tree size}. It is known that even in the
  two-valued base case, i.e.~in the standard modal logic of Kripke
  frames, the tree size of distinguishing formulae for non-bisimilar
  states is worst-case exponential~\cite{FigueiraGorin10}.
\end{rem}

\begin{rem}[Constants]\label{rem:constants-two-valued}
  In the interest of succinctness, we have allowed real-valued
  thresholds~$q$ on modalities~$\lambda_q$, following Desharnais et
  al.~\cite{DesharnaisEA08}. It is apparent from
  \Cref{alg:two-valued-extraction} that in case the input coalgebras
  have only rational values
  (i.e~$\lambda(A)(\xi(x)),\lambda(B)(\zeta(y))\in\Rat$ for all
  $x\in X$, $y\in Y$, $A\in 2^X$, $B\in 2^Y$), the extracted formulae
  contain only rational thresholds, which under mild assumptions on
  the given predicate liftings will moreover be of polynomial binary
  representation size; for instance, in our running example of the
  modality~$P$ over Markov chains, the modality will just sum over
  probabilities already given in the model representation.
\end{rem}
\noindent In combination with \cref{thm:strat}, we obtain
\begin{cor}[Polynomial-time computation of distinguishing formulae]
  \label{cor:compute-two-valued}
  If~$\Lambda$ is polynomial-time solvable, then
  $\epsilon$-distinguishing formulae of states $x,y$ of
  threshold-behavioural distance $\disteps{\Lambda}(x,y)>\epsilon$ can
  be computed in polynomial time.
\end{cor}
\begin{expl}
  \begin{enumerate}[wide]
  \item Continuing \Cref{expl:epsilon-distance}, we obtain Desharnais
    et al.'s expressiveness result for $\epsilon$-logic
    w.r.t.~$\epsilon$-simulation~\cite{DesharnaisEA08}, restricted to
    finitely branching labelled Markov chains, as an instance of
    \Cref{thm:hm-two-valued} (we leave a generalization of
    \Cref{thm:hm-two-valued} to countable branching as covered by
    Desharnais et al.~to future work). As a new result, we obtain by
    \Cref{thm:extract-two-valued,cor:compute-two-valued} and
    \Cref{expl:poly-solvable} that over finite systems,
    $\epsilon$-distinguishing formulae for both one-sided and
    two-sided L\'evy-Prokhorov distance on labelled Markov chains are
    computable in polynomial time and have quadratic modal
    rank. Further exploiting \Cref{expl:poly-solvable}, we obtain the
    same results for the following logics, in each case both for the
    asymmetric version and the symmetric version obtained by closing
    under duals:
  \item The logic of the modalities $\Diamond_{a,q}$ (indexed over
    labels~$a$ and thresholds~$q$) over generative probabilistic
    transition systems;
  \item a logic induced by threshold-indexed diamond
    modalities~$\Diamond_q$ over fuzzy transition systems; and
  \item a logic over metric transition systems featuring
    modalities~$\Diamond_{a,q}$ indexed over labels~$a$ (from a metric
    space) and thresholds~$q$.
  \end{enumerate}
  In these last cases, not only the respective results on computing
  distinguishing formulae but also the logics themselves and their
  expressiveness over finitely branching systems, i.e.~the respective
  instances of \Cref{thm:hm-two-valued}, appear to be new.
\end{expl}

\section{Quantitative Distinguishing Formulae}\label{sec:quant-log}

Having seen a characteristic two-valued modal logic for
threshold-based behavioural distance~$\disteps{\Lambda}$, we next
develop a properly quantitative modal logic from the given $2$-to-$\V$
predicate liftings. This logic will also be shown to be characteristic
for~$\disteps{\Lambda}$, now in the sense that behavioural distance
equals the supremum over all deviations in truth values of
quantitative modal formulae.

To this end, we now convert the given $2$-to-$\V$ predicate
liftings~$\lambda\in\Lambda$ into $\V$-to-$\V$ predicate liftings
$\gen{\lambda}$ (with components
$\gen{\lambda}_X\colon\V^X\to\V^{FX}$) -- the \emph{Sugeno modalities}
-- by putting
\begin{equation}\label{eq:sugeno-modality}
  \gen{\lambda}_X(f)(a)=\bigvee_{\epsilon\in\V}\epsilon\wedge\lambda(f_\epsilon)(a)
\end{equation}
where the $2$-valued predicate $f_\epsilon\colon X\to 2$ is given by
$f_\epsilon(x)=\top$ iff $f(x)\ge\epsilon$. We put
\begin{equation*}
  \Gen{\Lambda}=\{\gen{\lambda}\mid\lambda\in\Lambda\}.
\end{equation*}
We combine these modalities with a propositional base used standardly
in quantitative modal logics characterizing behavioural distances,
which besides minimum and maximum operators notably includes constant
shift
operators~\cite{BreugelWorrell05,wspk:van-benthem-fuzzy,KonigMikaMichalski18,WildSchroder22}. This
propositional base is designed to ensure non-expansiveness of formula
evaluation, and is correspondingly sometimes termed the
\emph{non-expansive propositional base}~\cite{GebhartEA25}.  We define
the set $\FLang(\Gen{\Lambda})$ of \emph{(quantitative)
  $\Lambda$-formulae}~$\phi,\psi$ by the grammar
\begin{equation*}
  \FLang(\Gen{\Lambda})\owns \phi,\psi ::= \top \mid \bot \mid \phi\land\psi \mid \phi\lor \psi  \mid \phi\oplus q \mid \phi\ominus q\mid
  \gen{\lambda}\phi
\end{equation*}
where $q\in\V$. The \emph{modal rank} of a formula~$\phi$ is
the maximal nesting depth of modalities~$\gen{\lambda}$ in~$\phi$.

Given a coalgebra $\xi\colon X\to FX$,
a formula~$\phi$ is interpreted as a truth map
$\Sem{\phi}\colon X\to \mathcal{V}$ given by the standard clauses for
the propositional operators ($\Sem{\top}(x)=1$, $\Sem{\bot}(x)=0$,
$\Sem{\phi\land \psi}(x) = \min\{\Sem{\phi}(x),\Sem{\psi}(x)\}$,
$\Sem{\phi\lor \psi}(x) = \max\{\Sem{\phi}(x),\Sem{\psi}(x)\}$,
$\Sem{\phi\oplus q}(x) = \Sem{\phi}(x)\oplus q$,
$\Sem{\phi\ominus q}(x) = \Sem{\phi}(x)\ominus q$ where $\oplus$,
$\ominus$ denote truncated addition and subtraction, respectively) and
by
\begin{equation*}
  \Sem{\gen{\lambda}\phi}(x) =
  \gen{\lambda}_X(\Sem{\phi})(\xi(x)).
\end{equation*}
We refer to the arising logic~$\Lang(\Gen{\Lambda})$ as
\emph{non-expansive coalgebraic modal logic}. We explicitly do not
include negation (which in the quantitative setting is interpreted as
$\Sem{\neg\phi}(x)=1\ominus\Sem{\phi}(x)$), as we also want to cover
asymmetric distances. In case~$\Lambda$ is closed under duals, we can
encode negation by taking negation normal forms
(cf.~\Cref{rem:sugeno-duals}).
\begin{expl}\label{expl:generally}
  For $F=\sfun$ and $\lambda=P$, the Sugeno modality~$S^P$, given by
  $S^P_X(f)(\mu)=\bigvee_{\epsilon\in\V}\epsilon\wedge\mu(f_\epsilon)$,
  is precisely the \emph{generally} modality used in probabilistic
  knowledge representation to model vague properties holding with high
  probability (more precisely, the original modality has only been
  considered on the subfunctor~$\dfun$ of~$\sfun$). As such, it serves
  as a computationally more tractable alternative to a modality
  \emph{probably} that coincides with the expectation
  modality~\cite{SchroderPattinson11}.
\end{expl}
\begin{rem}
  We use the term \emph{Sugeno modalities} and the
  designation~$\gen{\lambda}$ in honour of the fact
  that~\eqref{eq:sugeno-modality} essentially defines $\gen{\lambda}$
  using Sugeno integration~\cite{SugenoThesis}. In more detail, in the
  present discrete setting, a \emph{monotone measure} on a set~$X$ is
  a monotone function $g\colon\pfun(X)\to\V$ (such that additionally
  $g(\emptyset)=0$, a condition that is not needed here). Then, the
  \emph{Sugeno integral} $\dashint_X f(x) \circ g$ of a function
  $f\colon X\to\V$ over the set $X$ (in general, one can also
  integrate over subsets of~$X$ but this is not needed here) is
  \begin{equation*}
    \dashint_X f(x) \circ g = \bigvee_{\epsilon \in \V} ( \epsilon\wedge g(f_\epsilon)).
  \end{equation*}
  Thus, $\gen{\lambda}(f)(a)$ is just the Sugeno integral
  $\dashint_X f(x) \circ g$ where $g(A)=\lambda_X(A)(a)$.
\end{rem}
\noindent We have the following alternative description of
$\gen{\lambda}$:
\begin{lem}\label{lem:G-lambda-alt}
  For $f\in\V^X$, $a\in FX$, we have
  \begin{equation}\label{eq:G-lambda-set}
    \{\epsilon\wedge\lambda(f_\epsilon)(a)\mid\epsilon\ge 0\}=
    \{\epsilon\ge 0\mid \lambda(f_\epsilon)(a)\ge\epsilon\},
  \end{equation}
  and hence
  \begin{equation*}
    \gen{\lambda}(f)(a)=\bigvee\{\epsilon\ge 0\mid \lambda(f_\epsilon)(a)\ge\epsilon\}.
  \end{equation*}
\end{lem}
\begin{proof}
  `$\subseteq$': Let $\epsilon\ge 0$, and put
  $u=\epsilon\wedge\lambda(f_\epsilon)(a)$. Then $u\le\epsilon$, and
  therefore $\lambda(f_u)\ge\lambda(f_\epsilon)\ge u$.

  `$\supseteq$': Let $\epsilon\ge 0$ such that
  $\lambda(f_\epsilon)(a)\ge\epsilon$. Then
  $\epsilon=\epsilon\wedge\lambda(f_\epsilon)(a)$.
\end{proof}

\noindent
Additionally, we may alternatively express $\gen{\lambda}$ using infima:

\begin{lem}\label{lem:G-lambda-alt-inf}
  For $f\in\V^X$, $a\in FX$, we have
  \begin{equation*}
    \gen{\lambda}(f)(a)=\bigwedge\{\epsilon\ge 0\mid \lambda(f_\epsilon)(a)<\epsilon\}
    =\bigwedge_{\epsilon\ge 0} \epsilon\vee\lambda(f_\epsilon)(a).
  \end{equation*}
\end{lem}
\begin{proof}
  The two sets $U = \{\epsilon\ge 0\mid \lambda(f_\epsilon)(a)\ge\epsilon\}$ and $V = \{\epsilon\ge 0\mid \lambda(f_\epsilon)(a)<\epsilon\}$ are disjoint and their union is the unit interval $[0,1]$.
  As the map $\epsilon\mapsto\lambda(f_\epsilon)(a)$ is antimonotone by monotonicity of $\lambda$, we additionally have that $\epsilon < \epsilon'$ for every $\epsilon\in U, \epsilon'\in V$, as well as $\bigvee U = \bigwedge V$, establishing the first equality by \Cref{lem:G-lambda-alt}.

  The second equality is proved analogously to \Cref{lem:G-lambda-alt}.\pwnote{But still do it anyways}
\end{proof}

\begin{rem}\label{rem:sugeno-duals}
  The Sugeno construction preserves duals, in the sense that taking the Sugeno modality of the dual or a $2$-to-$\V$ predicate lifting is the same as taking the dual of its Sugeno modality.
  This is a consequence of \Cref{lem:G-lambda-alt-inf}; details are in the appendix.
  This means that if $\Lambda$ is closed under duals, then $\Gen{\Lambda}$ is also closed under duals, and we can therefore recursively encode negation.
  We use De Morgan's laws for the Boolean operators, put
  $\neg(\phi\oplus q) = (\neg\phi)\ominus q$ and
  $\neg(\phi\ominus q) = (\neg\phi)\oplus q$ for constant shifts,
  and use duals to negate modalities:
  \[ \neg \gen{\lambda}(\phi) = \gen{\overline{\lambda}}(\neg\phi). \]
\end{rem}

\begin{rem}
  We discuss briefly how the Sugeno modalities arise naturally from
  a  categorical perspective. It is well-known
	\pnnote{who to cite}
	that every partially ordered set is isomorphic to
	its set of principal ideals and
	that the principal ideals on a complete lattice correspond to
	maps from it to the two-element chain that
	send suprema to infima.
	\pnnote{here we assume that $\V$ is the unit interval equipped with the natural order}
	From this and the fact that $\set$ is Cartesian closed,
	we obtain a natural transformation
	$\sigma \colon \set(-,\V) \to \set(\V,\set(-,2))$
	whose $X$-component sends a function
	$f \colon X \to \V$ to the function that maps $\epsilon \in \V$ to $f_\epsilon$.
	Furthermore, the natural transformation $\sigma$ is a section%
	\pnnote{this corestricts to an iso as described in pedro/u.tex}
	of the natural transformation $\rho \colon \set(\V,\set(-,2)) \to \set(-,\V)$
	whose $X$-component sends a function $\phi \colon \V \to \set(X,2)$ to
	the function that maps $x \in X$ to $\bigvee \{\epsilon \in \V \mid \phi(\epsilon)(x) = \top\}$.
	This makes it easy to construct a $\V$-to-$\V$ predicate lifting from a
	$2$-to-$\V$ predicate lifting $\lambda$ and
	a natural transformation $\mu \colon \set(\V,-)^2 \to \set(\V,-)$
	as the following composite:
	\begin{center}
		\begin{tikzcd}
			\set(-,\V)  & \set(\V,\set(-,2)) \\
			            & \set(\V,\set(\V,\set(F-,2)) \\
			\set(F-,\V) & \set(\V,\set(F-,2)).
		\ar[from=1-1, to=1-2, "\sigma"]
		\ar[from=1-1, to=3-1]
		\ar[from=1-2, to=2-2, "\sigma_F \cdot \lambda \cdot -"]
		\ar[from=2-2, to=3-2, "\mu"]
		\ar[from=3-2, to=3-1, "\rho_F"]
		\end{tikzcd}
	\end{center}
	The functor $\set(\V,-)$ is part of a monad whose $X$-component of the multiplication
	sends $\phi \in \set(\V,\set(\V,X))$ to
	the map that assigns $v \in \V$ to $\phi(v)(v)$.
	The $\V$-to-$\V$ predicate lifting $\gen{\lambda}$ is induced by this multiplication as described above.
	Indeed, unravelling the composite
	yields the description of $\gen{\lambda}$ introduced in \Cref{lem:G-lambda-alt}.
\end{rem}
\noindent Over finite sets, the supremum defining~$\gen{\lambda}$ is
actually a maximum:
\begin{lem}\label{lem:G-lambda-max}
  Let~$X$ be a finite set, and let $f\in\V^X$, $a\in FX$; put
  $q=\gen{\lambda}_X(f)(a)$. Then the set~\eqref{eq:G-lambda-set}
  defining $\gen{\lambda}_X(f)(a)$ as its supremum has a greatest
  element,~$q$; that is, $\lambda_X(f_q)(a)\ge q$, and
  $\gen{\lambda}_X(f)(a)=q\wedge\lambda(f_q)(a)$.
\end{lem}
\begin{proof}
  Notice that the right hand side of~\eqref{eq:G-lambda-set} is
  clearly downwards closed, and hence contains all
  $\epsilon\in[0,q)$. Since~$X$ is finite, the descending chain of
  subsets $f_{\epsilon}\subseteq X$ indexed over $\epsilon\in[0,q)$
  becomes stationary, say at~$\epsilon_0$. Then for all
  $x\in f_{\epsilon_0}$, we have $x\in f_{\epsilon}$ for all
  $\epsilon\in [\epsilon_0,q)$, that is, $f(x)\ge\epsilon$. Thus,
  $f(x)\ge q$, i.e.~$x\in f_q$; this shows that
  $f_q=f_{\epsilon_0}$. We thus have
  $\lambda_X(f_q)(a)=\lambda_X(f_\epsilon)(a)\ge\epsilon$ for all
  $\epsilon\in[\epsilon_0,q)$, and hence $\lambda_X(f_q)(a)\ge q$.
\end{proof}

\noindent We now revisit the $2$-to-$\V$ predicate liftings from
\Cref{expl:L-Lambda} and describe the induced Sugeno modalities. In
some cases, we find that the Sugeno modalities coincide with known
modalities from the literature.

\begin{propn}
  \label{propn:sugeno-modalities}
  Using the $2$-to-$\V$ predicate liftings $\lambda$ from
  \Cref{expl:L-Lambda} we obtain the following Sugeno modalities
  $\gen{\lambda}$. In each case $f\colon X\to \mathcal{V}$.
  \begin{enumerate}[wide]
  \item\label{item:L-Sugeno-lp} For the case of the subdistribution
    functor and the $2$-to-$\V$ predicate lifting $P$ (with
    $P(A)(\mu)=\mu(A)$, for $A\in 2^X$, $\mu\in\mathcal{S}X$), we
    obtain the Sugeno modality
    \[ \gen{P}(f)(\mu) = \bigvee_{\epsilon\ge 0} \epsilon \land
      \mu(f_\epsilon). \]
  \item\label{item:L-Sugeno-pts} For the case of the functor
    $FX=\mathcal{D}(\mathcal{A}\times X)$ and the $2$-to-$\V$ predicate
    lifting $\Diamond_a$, we obtain the Sugeno modality
    \[ \gen{\Diamond_a}(f)(\mu) =
       \gen{P}(f)(\mu\big|_{\{a\}\times X}) \]
     for $\mu\in \mathcal{D}(\mathcal{A}\times X)$, where
     $\mu\big|_{\{a\}\times X}$ is the restriction of the
     distribution~$\mu$ to a subdistribution on the subset
     $\{a\}\times X$ and $\gen{P}$ is the Sugeno modality from
     (\ref{item:L-Sugeno-lp}).
  \item\label{item:L-Sugeno-fts} For the case of the fuzzy powerset
    functor and the $2$-to-$\V$ predicate
    lifting $\Diamond$, we obtain the Sugeno modality
    \[ \gen{\Diamond}(f)(g) = \bigvee_{x\in X} (g(x)\land f(x)) \]
    for $g\in \mathcal{V}^X$.
  \item\label{item:L-Sugeno-mts} \emph{Metric transition systems:} the
    $2$-to-$\V$ predicate lifting $\Diamond_a$ results in the Sugeno
    modality
    \[ \gen{\Diamond_a}(f)(S) = \bigvee_{(b,x)\in S} ((1-d(a,b))\land
      f(x)) \]
    for $S\subseteq \mathcal{A}\times X$.
  \item\label{item:L-Sugeno-convex} For the case of the convex
    powerset functor and the $2$-to-$\V$ predicate lifting~$\Diamond$,
    we obtain the Sugeno modality
    \[ \gen{\Diamond}(f)(U) = \bigvee_{\mu\in U} \gen{P}(f)(\mu) \]
    for $U\in \mathcal{C}X$, i.e., for $U$ being a convex set of
    distributions, where again~$\gen{P}$ is the Sugeno modality from
    (\ref{item:L-Sugeno-lp}).
  \end{enumerate}
\end{propn}

As already noted in \Cref{expl:generally}, the Sugeno modality from
Item~(\ref{item:L-Sugeno-lp}) is the \emph{generally} modality from
probabilistic knowledge representation. The Sugeno modality for fuzzy
powerset identified in Item~(\ref{item:L-Sugeno-fts}) coincides with
the usual fuzzy diamond modality, whose induced non-expansive modal
logic characterizes bisimulation distance on $\V$-valued relational
models~\cite{wspk:van-benthem-fuzzy}.

For metric transition systems (Item~(\ref{item:L-Sugeno-mts})), we
again obtain the usual metric diamond modality from the literature
whose induced non-expansive modal logic characterizes behavioural
metrics for metric transition
systems~\cite{afs:linear-branching-metrics,bgkm:hennessy-milner-galois,fswbgkm:quantitative-graded-semantics}.

\begin{defn}[Sugeno logical distance]
  Let $\xi\colon X\to FX$, $\zeta\colon Y\to FY$ be $F$-coalgebras.
  The \emph{($\Lambda$-)-Sugeno logical distance}
  $\logdist_{\Gen{\Lambda}}(c,d)$ from a state $c\in X$ to a state
  $d\in Y$ is
  \begin{equation*}
    \logdist_{\Gen{\Lambda}}(c,d)=\bigvee_{\phi\in\FLang(\Gen{\Lambda})}\Sem{\phi}(d)-\Sem{\phi}(c).
  \end{equation*}
\end{defn}
\noindent Notice that like threshold-based $\Lambda$-behavioural
distance, $\Lambda$-Sugeno distance is in general only a hemimetric,
i.e.~may fail to be symmetric; by \Cref{rem:sugeno-duals}, it is
symmetric, i.e.~a pseudometric, if~$\Lambda$ is closed under
duals. From \Cref{thm:lp-kantorovich-coalg} proved later, it follows
by general
results~\cite{KonigMikaMichalski18,WildSchroder21,ForsterEA23} that
the non-expansive coalgebraic modal logic $\Lang(\Gen{\Lambda})$ is
\emph{expressive}, i.e.~that $\Lambda$-Sugeno logical
distance~$\logdist_{\Gen{\Lambda}}$ coincides with threshold-based
$\Lambda$-behavioural distance~$\disteps{\Lambda}$. For simplicity, we
do opt to base the inequality
$\logdist_{\Gen{\Lambda}}\le \disteps{\Lambda}$ on
\Cref{thm:lp-kantorovich-coalg} and the mentioned general results; for
the (harder) inequality
$\logdist_{\Gen{\Lambda}}\ge \disteps{\Lambda}$, we give a direct (and
rather simpler proof) for the present case in the appendix, obtaining

\begin{thm}[Sugeno expressivity]\label{thm:expr-sugeno}
  $\Lambda$-Sugeno logical distance $\logdist_{\Gen{\Lambda}}$ and
  $\laxv_\Lambda$-behavioural distance $\disteps{\Lambda}$ coincide on
  finitely branching $F$-coalgebras.
\end{thm}
\noindent In the following, we focus mainly on the algorithmic
construction of distinguishing formulae, which now witness lower
bounds for behavioural distance directly by formula evaluation:
\begin{defn}
  Let $x\in X$, $y\in Y$ be states in $F$-coalgebras
  $\xi\colon X\to FX$, $\zeta\colon Y\to FY$, respectively. A
  quantitative formula $\phi\in\FLang(\Gen{\Lambda})$ is
  \emph{$\epsilon$-distinguishing} for $(x,y)$ if
  $\Sem{\phi}(y)<\Sem{\phi}(x)-\epsilon$.
\end{defn}
By \Cref{thm:expr-sugeno}, if~$\xi$ and~$\zeta$ are finitely
branching, then an $\epsilon$-distinguishing formula for~$(x,y)$
exists whenever $\disteps{\Lambda}(x,y)>\epsilon$. Over finite
coalgebras, we again obtain distinguishing formulae from winning
strategies of Spoiler in the threshold-quantitative codensity game,
proceeding in a similar fashion as in the two-valued case
(\Cref{sec:two-valued-logic}). Specifically, let~$\xi\colon X\to FX$,
$\zeta\colon Y\to FY$ be finite $F$-coalgebras, and let~$s$ be a
memoryless winning strategy for Spoiler that wins all positions in the
winning region of Spoiler in the codensity game up to~$\epsilon$
on~$\xi$ and~$\zeta$ as per \cref{def:codensity-game}, noting as
before that Spoiler wins any position in their winning region in at
most~$k=|X|\times|Y|$ moves (but at least one). We compute
$\epsilon$-distinguishing formulae for all pairs $(x,y)$ of states in
Spoiler's winning region in dag representation using dynamic
programming, with stages indexed by the remaining number of Spoiler
moves. We give the algorithm for stage $i>0$, given a position
$(x_0,y_0)$ won by~$S$ in~$i$ moves:

\begin{alg}[Extraction of distinguishing
  formulae]\label{alg:extraction} ~

  \begin{enumerate}
  \item Put $(A,B)=s(x_0,y_0)$ (Spoiler's winning move). By definition
    of the game, there exists $\lambda\in\Lambda$ such that
    $\lambda(B)(\zeta(y))<q-\epsilon$ where $q=\lambda(A)(\xi(x))$.
  \item For every possible reply by~$D$, i.e.~every pair
    $(x,y)\in X\times Y$ such that $x\in A$ and $y\notin B$,~$S$ wins
    $(x,y)$ in at most $i-1$ moves, so we have already computed an
    $\epsilon$-distinguishing formula $\phi_{xy}$ for $(x,y)$ in
    previous stages. By applying constant shifts, we normalize
    the~$\phi_{x,y}$ so that $\Sem{\phi_{xy}}(x)=q$.
  \item Put
    \begin{equation}
      \label{eq:distinguishing-phi}
      \textstyle \phi=\bigvee_{x\in A}\bigwedge_{y\in Y\setminus B}\,\phi_{xy}.
    \end{equation}
  \item Assign the $\epsilon$-distinguishing formula
    $\phi_{x_0y_0}=\gen{\lambda}\phi$ to $(x_0,y_0)$.
  \end{enumerate}
\end{alg}
\begin{rem}[Constants in shifts]
  Like in the two-valued case, we have admitted real-valued shifts in
  the grammar but note that the above algorithm uses only rational
  shifts provided that the data in the given models is rational in the
  same sense as in \Cref{rem:constants-two-valued}.
\end{rem}

\noindent The correctness of the algorithm and the complexity of the
formulae it computes are formulated as follows:
\begin{thm}\label{thm:extract-quant}
  Given finite $F$-coalgebras $\xi\colon X\to FX$,
  $\zeta\colon Y\to FY$ and a Spoiler strategy~$s$ for the codensity
  game up to~$\epsilon$ on~$\xi$ and~$\zeta$ that wins all positions
  in Spoiler's winning region, \cref{alg:extraction} constructs
  $\epsilon$-distinguishing formulae~$\phi_{xy}$ for all pairs of
  states $(x,y)\in X\times Y$ such that
  $\disteps{\Lambda}(x,y)>\epsilon$. The algorithm runs in polynomial
  time, and the formulae~$\phi_{x,y}$ it constructs are of quadratic
  modal rank and polynomial, specifically quartic, dag size.
\end{thm}
\noindent The proof, fairly similar to that of
\Cref{thm:extract-two-valued}, is in the appendix.
\begin{cor}[Polynomial-time computation of distinguishing quantitative
  formulae]
  \label{cor:compute-quant}
  If~$\Lambda$ is polynomial-time solvable, then quantitative
  $\epsilon$-distinguishing formulae of states $x,y$ of
  threshold-behavioural distance $\disteps{\Lambda}(x,y)>\epsilon$ can
  be computed in polynomial time.
\end{cor}

\begin{expl}
  We have shown in ~\Cref{expl:poly-solvable} that the modalities of Examples~\ref{expl:L-Lambda}.\ref{item:L-Lambda-lp}-\ref{item:L-Lambda-mts} are polynomial-time solvable.
  By~\Cref{cor:compute-quant} we can therefore compute quantitative $\epsilon$-distinguishing formulae for the following quantitative logics in polynomial time:
  \begin{enumerate}[wide]
    \item The logic of \emph{generally} for Markov chains, which is the characteristic logic for the Lévy-Prokhorov distance.
    \item The logic for generative probabilistic transition systems that features a copy of the generally modality for every label.
    \item The logic of \emph{fuzzy diamond}, i.e.~standard fuzzy modal
      logic
      (e.g.~\cite{Straccia98,wspk:van-benthem-fuzzy,GebhartEA25}) for
      fuzzy transition systems, i.e.~coalgebras for the fuzzy powerset
      functor~$\pfun_\V$.
    \item The logic for metric transition systems that features one
      modality~$\Diamond_a$ per label~$a$, whose value depends on the
      transition most closely matching the given label.
  \end{enumerate}
  All of these results appear to be new. The last case (metric
  transition systems) is similar to a version of metric transition
  systems with labels on the states (rather than the transitions), for
  which polynomial-time extraction of distinguishing formulae is known
  from previous work~\cite{afs:linear-branching-metrics} (a result
  that is also easily cast as an instance of
  \Cref{cor:compute-quant}).
\end{expl}

\section{Kantorovich Characterization of Coalgebraic  $\epsilon$-Bisimulation}\label{sec:kant}

As our last main result, we show that the so-called Kantorovich
quantitative (sometimes `fuzzy') lax extension~\cite{WildSchroder22}
induced by the Sugeno modalities coincides with the quantitative lax
extension~$\laxv_\Lambda$ as per~\eqref{eq:laxv}. This closes a number
of explicit open ends; in particular, it follows by general
results~\cite{WildSchroder22} that~$\laxv_\Lambda$ is indeed a
quantitative lax extension, and that $\Lambda$-Sugeno logical distance
is below threshold-based behavioural distance. We recall the
definition of the Kantorovich quantitative lax extension $K_\Theta$
of~$F$ induced by a set~$\Theta$ of $\V$-valued predicate liftings
for~$F$. Given a $\V$-valued relation $r\colon\frel{X}{Y}$, a pair
$(f,g)$ of predicates $f\colon X\to\V$, $g\colon Y\to\V$ is
\emph{$r$-preserved} if for all $(x,y)\in X\times Y$,
$g(y)\ge f(x)-r(x,y)$ (equivalently $f(x)-g(y)\le r(x,y)$). For a
$\V$-valued relation $r\colon\frel{X}{Y}$, and $a\in FX$, $b\in FY$,
we then put
\begin{equation*}
  \kantv_\Theta r(a,b)=\bigvee\{\mu(f)(a)\ominus\mu(g)(b)\mid \mu\in\Theta, (f,g)\text{ $r$-preserved}\}.
\end{equation*}
It has been shown that $\kantv_\Theta$ is always a quantitative lax
extension~\cite{WildSchroder22} and that the quantitative modal logic
induced by~$\Theta$ is characteristic for the behavioural distance
induced by~$\kantv_\Theta$. The following theorem implies that these
results apply to the Sugeno modalities and~$\laxv_\Lambda$:
\begin{thm}\label{thm:lp-kantorovich-coalg}
  For~$\laxv_\Lambda$ as per~\eqref{eq:laxv} and $r\colon\frel{X}{Y}$,
  we have
  \begin{equation*}
    \laxv_\Lambda r=\kantv_{\Gen{\Lambda}}r.
  \end{equation*}
\end{thm}

\begin{expl}
  As discussed after \Cref{propn:sugeno-modalities}, the obtained
  Sugeno modalities coincide with known quantitative modalities and
  hence with the Kantorovich characterizations in the literature
  \cite{WildEA25,wspk:van-benthem-fuzzy,afs:linear-branching-metrics,bgkm:hennessy-milner-galois,fswbgkm:quantitative-graded-semantics}
  and whenever we use the same propositional operators in the logic,
  the corresponding $\epsilon$-distances witness the known behavioural
  metrics.
  \begin{enumerate}[wide]
    \item We obtain a Kantorovich characterization of Lévy-Prokhorov distance through the \emph{generally} modality.
      Such a characterization has been obtained quite recently in the probabilistic (rather than subprobabilistic) setting~\cite{WildEA25}; this result can easily be recovered using the following item.
    \item We obtain a Kantorovich characterization of a distance on probabilistic transition systems that behaves like the Lévy-Prokhorov distance on each slice corresponding to a label.
      Specifically, we assign to $\mu\in\dfun(\Act\times X)$ and $\nu\in\dfun(\Act\times Y)$ the distance
      \[ \bigvee_{a\in\Act} \laxv_{P}(\mu\big|_{\{a\}\times X},
      \nu\big|_{\{a\}\times Y}), \]
      where $\laxv_{P}$ is Lévy-Prokhorov distance and $\mu\big|_{\{a\}\times X}$, $\nu\big|_{\{a\}\times Y}$ are as in ~\Cref{propn:sugeno-modalities}.
      For $|\Act| = 1$ the functor $\dfun(\Act\times-)$ is isomorphic to $\dfun$, so that we recover precisely the Kantorovich characterization of Lévy-Prokhorov for distributions~\cite{WildEA25}.
    \item For fuzzy transition systems, we obtain the Kantorovich distance for the \emph{fuzzy diamond} modality from fuzzy modal logic~\cite{wspk:van-benthem-fuzzy}.
  \end{enumerate}
\end{expl}


\section{Conclusion}
\label{sec:conclusion}

We have introduced a coalgebraic framework for behavioural distances
that are \emph{threshold-based}, i.e.~determined by notions of
$\epsilon$-\mbox{(bi-)}simulations working with a fixed allowed
deviation~$\epsilon$. Our framework is parametrized over a choice of
\emph{$2$-to-$\V$ predicate liftings} that lift two-valued predicates
to quantitative predicates (taking truth values in the unit
interval~$\V$), in generalization of the probability operator. We have
shown that such distances are equivalently induced by suitably
constructed quantitative lax extensions, which in turn we have shown
to be characterized by quantitative modalities induced from the given
$2$-to-$\V$ predicate liftings, the \emph{Sugeno modalities}. Both for
a two-valued modal logic determined by the original $2$-to-$\V$
predicate liftings and for a quantitative modal logic determined by
the Sugeno modalities, we have established quantitative
Hennessy-Milner theorems stating that the respective notions of
logical distance coincide with threshold-based behavioural distance on
finitely branching systems (going via an existing quantitative
coalgebraic Hennessy-Milner
theorem~\cite{KonigMikaMichalski18,WildSchroder22} in the quantitative
case). Moreover, we have shown that distinguishing formulae in this
logic, which certify that the behavioural distance of two given states
is above a given threshold, have polynomial dag size and can, under
mild conditions, be computed in polynomial time. These results subsume
established as well as recent results on the L\'evy-Prokhorov distance
on Markov chains, such as the characterization via a quantitative lax
extension~\cite{DesharnaisSokolova25} and logical characterizations
via both two-valued~\cite{DesharnaisEA08} and
quantitative~\cite{WildEA25} modal logics, as well as a result on
polynomial-time computation of distinguishing formulae for metric
transition systems~\cite{afs:linear-branching-metrics}. All other
instances of our generic algorithmic results are, to our best
knowledge, new. Maybe most notably, we establish the polynomial-time
computation of distinguishing formulae, both two-valued and
quantitative, for L\'evy-Prokhorov behavioural distance, also known as
$\epsilon$-distance, on (labelled) Markov
chains~\cite{DesharnaisEA08}.

One important direction for future work is to develop generic
reasoning algorithms for both two-valued and quantitative
characteristic modal logics induced by $2$-to-$\V$ predicate
liftings. A possible model for such algorithms is provided in recent
work on reasoning in what has been termed \emph{non-expansive fuzzy
  $\ALC$}, effectively a multi-modal version of the modal logic of
fuzzy transition systems that has featured as one of our running
examples. Also, future investigations will be directed at going beyond
the setting of behavioural distances, covering, for instance,
behavioural (quasi-)uniformities~\cite{KomoridaEA21} (which, for
instance, play a role in GSOS formats for probabilistic
systems~\cite{GeblerEA16}) or general behavioural
relations~\cite{FordThesis,ForsterEA25}.

\bibliographystyle{ACM-Reference-Format}
\citestyle{acmnumeric}
\bibliography{references}

@PREAMBLE{ {\providecommand{\noopsort}[1]{}} }

@string{lncs="LNCS"}

@string{springer="Springer"}

@string{lipics="LIPIcs"}

@string{dagstuhl="Schloss Dagstuhl -- Leibniz-Zentrum f{\"u}r Informatik"}

@Article{afs:linear-branching-metrics,
  author       = {Luca de Alfaro and
                  Marco Faella and
                  Mari{\"{e}}lle Stoelinga},
  title        = {Linear and Branching System Metrics},
  journal      = {{IEEE} Trans.\ Software Eng.},
  volume       = {35},
  number       = {2},
  pages        = {258--273},
  year         = {2009},
  nourl          = {https://doi.org/10.1109/TSE.2008.106},
  doi          = {10.1109/TSE.2008.106},
  bibsource    = {dblp computer science bibliography, https://dblp.org}
}

@article{BreugelWorrell05,
  author       = {Franck van Breugel and
                  James Worrell},
  title        = {A behavioural pseudometric for probabilistic transition systems},
  journal      = {Theor.\ Comput.\ Sci.},
  volume       = {331},
  number       = {1},
  pages        = {115--142},
  year         = {2005},
  nourl          = {https://doi.org/10.1016/j.tcs.2004.09.035},
  doi          = {10.1016/J.TCS.2004.09.035},
  bibsource    = {dblp computer science bibliography, https://dblp.org}
}

@Article{bbkk:coalgebraic-behavioral-metrics,
  author       = {Paolo Baldan and
                  Filippo Bonchi and
                  Henning Kerstan and
                  Barbara K{\"{o}}nig},
  title        = {Coalgebraic Behavioral Metrics},
  journal      = {Log.\ Methods Comput.\ Sci.},
  volume       = {14},
  number       = {3},
  year         = {2018},
  nourl          = {https://doi.org/10.23638/LMCS-14(3:20)2018},
  doi          = {10.23638/LMCS-14(3:20)2018},
  bibsource    = {dblp computer science bibliography, https://dblp.org}
}

@InProceedings{bgkm:hennessy-milner-galois,
  author = 	 {Harsh Beohar and Sebastian Gurke and  Barbara K\"onig and
                  Karla Messing},
  title = 	 {{H}ennessy-{M}ilner Theorems via {G}alois Connections},
  booktitle =    {Proc. of CSL '23},
  publisher =	 {Schloss Dagstuhl -- Leibniz Center for Informatics},
  series = 	 {{LIPIcs}},
  year =         2023,
  note =         {to appear}
}

@InProceedings{c:automatically-explaining-bisim,
  author       = {Rance Cleaveland},
  editor       = {Edmund M. Clarke and
                  Robert P. Kurshan},
  title        = {On Automatically Explaining Bisimulation Inequivalence},
  ALTbooktitle    = {Computer Aided Verification, 2nd International Workshop, {CAV} '90,
                  New Brunswick, NJ, USA, June 18-21, 1990, Proceedings},
  booktitle    = {Computer Aided Verification,  {CAV} 1990},
  series       = lncs,
  volume       = {531},
  pages        = {364--372},
  publisher    = springer,
  year         = {1990},
  nourl          = {https://doi.org/10.1007/BFb0023750},
  doi          = {10.1007/BFB0023750},
  bibsource    = {dblp computer science bibliography, https://dblp.org}
}

@inproceedings{rb:explainability-labelled-mc,
  author       = {Amgad Rady and
                  Franck van Breugel},
  editor       = {Orna Kupferman and
                  Pawel Sobocinski},
  title        = {Explainability of Probabilistic Bisimilarity Distances for Labelled
                  Markov Chains},
  ALTbooktitle    = {Foundations of Software Science and Computation Structures - 26th
                  International Conference, FoSSaCS 2023, Held as Part of the European
                  Joint Conferences on Theory and Practice of Software, {ETAPS} 2023,
                  Paris, France, April 22-27, 2023, Proceedings},
  booktitle    = {Foundations of Software Science and Computation Structures, FoSSaCS 2023},
  series       = lncs,
  volume       = {13992},
  pages        = {285--307},
  publisher    = springer,
  year         = {2023},
  nourl          = {https://doi.org/10.1007/978-3-031-30829-1\_14},
  doi          = {10.1007/978-3-031-30829-1\_14},
  bibsource    = {dblp computer science bibliography, https://dblp.org}
}

@InProceedings{kms:non-bisimilarity-coalgebraic,
  author       = {Barbara K{\"{o}}nig and
                  Christina Mika{-}Michalski and
                  Lutz Schr{\"{o}}der},
  editor       = {Daniela Petrisan and
                  Jurriaan Rot},
  title        = {Explaining Non-bisimilarity in a Coalgebraic Approach: Games and Distinguishing
                  Formulas},
  ALTbooktitle    = {Coalgebraic Methods in Computer Science - 15th {IFIP} {WG} 1.3 International
                  Workshop, {CMCS} 2020, Colocated with {ETAPS} 2020, Dublin, Ireland,
                  April 25-26, 2020, Proceedings},
  booktitle    = {Coalgebraic Methods in Computer Science, {CMCS} 2020},
  series       = lncs,
  volume       = {12094},
  pages        = {133--154},
  publisher    = springer,
  year         = {2020},
  nourl          = {https://doi.org/10.1007/978-3-030-57201-3\_8},
  doi          = {10.1007/978-3-030-57201-3\_8},
  bibsource    = {dblp computer science bibliography, https://dblp.org}
}

@InProceedings{DesharnaisSokolova25,
  author       = {Jos{\'{e}}e Desharnais and
                  Ana Sokolova},
  title        = {{\(\epsilon\)}-Distance via L{\'{e}}vy-Prokhorov Lifting},
  year         = 2026,
  editor = {Stefano Guerrini and Barbara König},
  booktitle = {Computer Science Logic, CSL 2026},
  series = lipics,
  publisher = dagstuhl,
  note = {To appear; available on arXiv under {\url{https://arxiv.org/abs/2507.10732}}},
}

@inproceedings{DesharnaisEA08,
  author       = {Jos{\'{e}}e Desharnais and
                  Fran{\c{c}}ois Laviolette and
                  Mathieu Tracol},
  title        = {Approximate Analysis of Probabilistic Processes: Logic, Simulation
                  and Games},
  ALTbooktitle    = {Fifth International Conference on the Quantitative Evaluaiton of Systems
                  {(QEST} 2008), 14-17 September 2008, Saint-Malo, France},
  booktitle    = {Quantitative Evaluaiton of Systems
                  {(QEST} 2008)},
  pages        = {264--273},
  publisher    = {{IEEE} Computer Society},
  year         = {2008},
  nourl          = {https://doi.org/10.1109/QEST.2008.42},
  doi          = {10.1109/QEST.2008.42},
  bibsource    = {dblp computer science bibliography, https://dblp.org}
}

@inproceedings{NoraEA25,
  author       = {Pedro Nora and
                  Jurriaan Rot and
                  Lutz Schr{\"{o}}der and
                  Paul Wild},
  editor       = {Parosh Aziz Abdulla and
                  Delia Kesner},
  title        = {Relational Connectors and Heterogeneous Simulations},
  ALTbooktitle    = {Foundations of Software Science and Computation Structures - 28th
                  International Conference, FoSSaCS 2025, Held as Part of the International
                  Joint Conferences on Theory and Practice of Software, {ETAPS} 2025,
                  Hamilton, ON, Canada, May 3-8, 2025, Proceedings},
  booktitle    = {Foundations of Software Science and Computation Structures, FoSSaCS 2025},
  series       = lncs,
  volume       = {15691},
  pages        = {111--132},
  publisher    = springer,
  year         = {2025},
  nourl          = {https://doi.org/10.1007/978-3-031-90897-2\_6},
  doi          = {10.1007/978-3-031-90897-2\_6},
  bibsource    = {dblp computer science bibliography, https://dblp.org}
}

@inproceedings{SchroderPattinson11,
  author       = {Lutz Schr{\"{o}}der and
                  Dirk Pattinson},
  editor       = {Toby Walsh},
  title        = {Description Logics and Fuzzy Probability},
  ALTbooktitle    = {{IJCAI} 2011, Proceedings of the 22nd International Joint Conference
                  on Artificial Intelligence, Barcelona, Catalonia, Spain, July 16-22,
                  2011},
  booktitle    = {International Joint Conference
                  on Artificial Intelligence, IJCAI 2011},
  pages        = {1075--1081},
  publisher    = {{IJCAI/AAAI}},
  year         = {2011},
  nourl          = {https://doi.org/10.5591/978-1-57735-516-8/IJCAI11-184},
  doi          = {10.5591/978-1-57735-516-8/IJCAI11-184},
  bibsource    = {dblp computer science bibliography, https://dblp.org}
}

@inproceedings{WildSchroder21,
  author       = {Paul Wild and
                  Lutz Schr{\"{o}}der},
  editor       = {Stefan Kiefer and
                  Christine Tasson},
  title        = {A Quantified Coalgebraic van Benthem Theorem},
  ALTbooktitle    = {Foundations of Software Science and Computation Structures - 24th
                  International Conference, {FOSSACS} 2021, Held as Part of the European
                  Joint Conferences on Theory and Practice of Software, {ETAPS} 2021,
                  Luxembourg City, Luxembourg, March 27 - April 1, 2021, Proceedings},
  booktitle    = {Foundations of Software Science and Computation Structures, {FOSSACS} 2021},
  series       = lncs,
  volume       = {12650},
  pages        = {551--571},
  publisher    = springer,
  year         = {2021},
  nourl          = {https://doi.org/10.1007/978-3-030-71995-1\_28},
  doi          = {10.1007/978-3-030-71995-1\_28},
  bibsource    = {dblp computer science bibliography, https://dblp.org}
}

@inproceedings{ForsterEA23,
  author       = {Jonas Forster and
                  Sergey Goncharov and
                  Dirk Hofmann and
                  Pedro Nora and
                  Lutz Schr{\"{o}}der and
                  Paul Wild},
  editor       = {Bartek Klin and
                  Elaine Pimentel},
  title        = {Quantitative Hennessy-Milner Theorems via Notions of Density},
  ALTbooktitle    = {31st {EACSL} Annual Conference on Computer Science Logic, {CSL} 2023,
                  Warsaw, Poland, February 13-16, 2023},
  booktitle    = {Computer Science Logic, {CSL} 2023},
  series       = {LIPIcs},
  volume       = {252},
  pages        = {22:1--22:20},
  publisher    = {Schloss Dagstuhl -- Leibniz-Zentrum f{\"{u}}r Informatik},
  year         = {2023},
  nourl          = {https://doi.org/10.4230/LIPIcs.CSL.2023.22},
  doi          = {10.4230/LIPICS.CSL.2023.22},
  bibsource    = {dblp computer science bibliography, https://dblp.org}
}

@inproceedings{KonigMikaMichalski18,
  author       = {Barbara K{\"{o}}nig and
                  Christina Mika{-}Michalski},
  editor       = {Sven Schewe and
                  Lijun Zhang},
  title        = {(Metric) Bisimulation Games and Real-Valued Modal Logics for Coalgebras},
  ALTbooktitle    = {29th International Conference on Concurrency Theory, {CONCUR} 2018,
                  Beijing, China, September 4-7, 2018},
  booktitle    = {Concurrency Theory, {CONCUR} 2018},
  series       = {LIPIcs},
  volume       = {118},
  pages        = {37:1--37:17},
  publisher    = {Schloss Dagstuhl -- Leibniz-Zentrum f{\"{u}}r Informatik},
  year         = {2018},
  nourl          = {https://doi.org/10.4230/LIPIcs.CONCUR.2018.37},
  doi          = {10.4230/LIPICS.CONCUR.2018.37},
  bibsource    = {dblp computer science bibliography, https://dblp.org}
}

@inproceedings{GiacaloneEA90,
  author       = {Alessandro Giacalone and
                  Chi{-}Chang Jou and
                  Scott A. Smolka},
  editor       = {Manfred Broy and
                  Cliff B. Jones},
  title        = {Algebraic Reasoning for Probabilistic Concurrent Systems},
  ALTbooktitle    = {Programming concepts and methods: Proceedings of the {IFIP} Working
                  Group 2.2, 2.3 Working Conference on Programming Concepts and Methods,
                  Sea of Galilee, Israel, 2-5 April, 1990},
  booktitle    = {Programming concepts and methods},
  pages        = {443--458},
  publisher    = {North-Holland},
  year         = {1990},
  bibsource    = {dblp computer science bibliography, https://dblp.org}
}

@article{Rutten00,
  author       = {Jan J. M. M. Rutten},
  title        = {Universal coalgebra: a theory of systems},
  journal      = {Theor.\ Comput.\ Sci.},
  volume       = {249},
  number       = {1},
  pages        = {3--80},
  year         = {2000},
  nourl          = {https://doi.org/10.1016/S0304-3975(00)00056-6},
  doi          = {10.1016/S0304-3975(00)00056-6},
  bibsource    = {dblp computer science bibliography, https://dblp.org}
}

@article{WildSchroder22,
  author       = {Paul Wild and
                  Lutz Schr{\"{o}}der},
  title        = {Characteristic Logics for Behavioural Hemimetrics via Fuzzy Lax Extensions},
  journal      = {Log.\ Methods Comput.\ Sci.},
  volume       = {18},
  number       = {2},
  year         = {2022},
  nourl          = {https://doi.org/10.46298/lmcs-18(2:19)2022},
  doi          = {10.46298/LMCS-18(2:19)2022},
  bibsource    = {dblp computer science bibliography, https://dblp.org}
}

@phdthesis{SugenoThesis,
  title={Theory of fuzzy integrals and its applications},
  author={Sugeno, Michio},
  school={Tokyo Institute of Technology},
  year={1974}
}

@InProceedings{WildEA25,
  title={Generalized Kantorovich-Rubinstein Duality beyond Hausdorff and Kantorovich}, 
  author={Paul Wild and Lutz Schröder and Karla Messing and Barbara König and Jonas Forster},
  year={2026},
  editor = {Natalie Betrand and Stefan Milius},
  booktitle = {Foundations of Software Science and Computation Structures, FoSSaCS 2026},
  series = lncs,
  publisher = springer,
  note = {To appear; available on arXiv under {\url{https://arxiv.org/abs/2510.23552}}},
}

@inproceedings{KomoridaEA19,
  author       = {Yuichi Komorida and
                  Shin{-}ya Katsumata and
                  Nick Hu and
                  Bartek Klin and
                  Ichiro Hasuo},
  title        = {Codensity Games for Bisimilarity},
  ALTbooktitle    = {34th Annual {ACM/IEEE} Symposium on Logic in Computer Science, {LICS}
                  2019, Vancouver, BC, Canada, June 24-27, 2019},
  booktitle    = {Logic in Computer Science, {LICS}
                  2019},
  pages        = {1--13},
  publisher    = {{IEEE}},
  year         = {2019},
  nourl          = {https://doi.org/10.1109/LICS.2019.8785691},
  doi          = {10.1109/LICS.2019.8785691},
  bibsource    = {dblp computer science bibliography, https://dblp.org}
}

@inproceedings{FigueiraGorin10,
  author       = {Santiago Figueira and
                  Daniel Gor{\'{\i}}n},
  editor       = {Lev D. Beklemishev and
                  Valentin Goranko and
                  Valentin B. Shehtman},
  title        = {On the Size of Shortest Modal Descriptions},
  ALTbooktitle    = {Advances in Modal Logic 8, papers from the eighth conference on "Advances
                  in Modal Logic," held in Moscow, Russia, 24-27 August 2010},
  booktitle    = {Advances in Modal Logic, AiML 2010},
  pages        = {120--139},
  publisher    = {College Publications},
  year         = {2010},
  url          = {http://www.aiml.net/volumes/volume8/Figueira-Gorin.pdf},
  bibsource    = {dblp computer science bibliography, https://dblp.org}
}

@article{WissmannEA22,
  author       = {Thorsten Wi{\ss}mann and
                  Stefan Milius and
                  Lutz Schr{\"{o}}der},
  title        = {Quasilinear-time Computation of Generic Modal Witnesses for Behavioural
                  Inequivalence},
  journal      = {Log.\ Methods Comput.\ Sci.},
  volume       = {18},
  number       = {4},
  year         = {2022},
  url          = {https://doi.org/10.46298/lmcs-18(4:6)2022},
  doi          = {10.46298/LMCS-18(4:6)2022},
  bibsource    = {dblp computer science bibliography, https://dblp.org}
}

@article{FahrenbergLegay14,
  author       = {Uli Fahrenberg and
                  Axel Legay},
  title        = {The quantitative linear-time-branching-time spectrum},
  journal      = {Theor.\ Comput.\ Sci.},
  volume       = {538},
  pages        = {54--69},
  year         = {2014},
  nourl        = {https://doi.org/10.1016/j.tcs.2013.07.030},
  doi          = {10.1016/J.TCS.2013.07.030},
  bibsource    = {dblp computer science bibliography, https://dblp.org}
}

@article{GeblerEA16,
  author       = {Daniel Gebler and
                  Kim G. Larsen and
                  Simone Tini},
  title        = {Compositional bisimulation metric reasoning with Probabilistic Process
                  Calculi},
  journal      = {Log.\ Methods Comput.\ Sci.},
  volume       = {12},
  number       = {4},
  year         = {2016},
  nourl          = {https://doi.org/10.2168/LMCS-12(4:12)2016},
  doi          = {10.2168/LMCS-12(4:12)2016},
  bibsource    = {dblp computer science bibliography, https://dblp.org}
}

@phdthesis{FordThesis,
  author       = {Chase Ford},
  title        = {Presentations of Graded Coalgebraic Semantics},
  school       = {Friedrich-Alexander-Universität Erlangen-Nürnberg},
  year         = {2023},
  url          = {https://opus4.kobv.de/opus4-fau/frontdoor/index/index/docId/23568},
  urn          = {urn:nbn:de:bvb:29-opus4-235684},
  bibsource    = {dblp computer science bibliography, https://dblp.org}
}

@inproceedings{ForsterEA25,
  author       = {Jonas Forster and
                  Lutz Schr{\"{o}}der and
                  Paul Wild},
  title        = {Conformance Games for Graded Semantics},
  ALTbooktitle    = {40th Annual {ACM/IEEE} Symposium on Logic in Computer Science, {LICS}
                  2025, Singapore, June 23-26, 2025},
 booktitle    = {Logic in Computer Science, {LICS}
                  2025},
  pages        = {555--567},
  publisher    = {{IEEE}},
  year         = {2025},
  nourl          = {https://doi.org/10.1109/LICS65433.2025.00048},
  doi          = {10.1109/LICS65433.2025.00048},
  bibsource    = {dblp computer science bibliography, https://dblp.org}
}

@inproceedings{KomoridaEA21,
  author       = {Yuichi Komorida and
                  Shin{-}ya Katsumata and
                  Clemens Kupke and
                  Jurriaan Rot and
                  Ichiro Hasuo},
  title        = {Expressivity of Quantitative Modal Logics : Categorical Foundations
                  via Codensity and Approximation},
  ALTbooktitle    = {36th Annual {ACM/IEEE} Symposium on Logic in Computer Science, {LICS}
                  2021, Rome, Italy, June 29 - July 2, 2021},
  booktitle    = {Logic in Computer Science, {LICS}
                  2021},
  pages        = {1--14},
  publisher    = {{IEEE}},
  year         = {2021},
  nourl          = {https://doi.org/10.1109/LICS52264.2021.9470656},
  doi          = {10.1109/LICS52264.2021.9470656},
  bibsource    = {dblp computer science bibliography, https://dblp.org}
}

@inproceedings{AolariteiEA25,
  author          = {Liviu Aolaritei and Oliver Wang and Julie Zhu and Michael Jordan and Youssef Marzouk},
  title           = {Conformal Prediction under {L}évy-{P}rokhorov Distribution Shifts: Robustness to Local and Global Perturbations},
  booktitle       = {Neural Information Processing Systems, NeurIPS 2025},
  year            = {2025},
  note            = {To appear. Preprint available at https://arxiv.org/abs/2502.14105}
 }

@inproceedings{BennounaEA23,
  author       = {M. Amine Bennouna and
                  Ryan Lucas and
                  Bart P. G. Van Parys},
  editor       = {Andreas Krause and
                  Emma Brunskill and
                  Kyunghyun Cho and
                  Barbara Engelhardt and
                  Sivan Sabato and
                  Jonathan Scarlett},
  title        = {Certified Robust Neural Networks: Generalization and Corruption Resistance},
  altbooktitle = {International Conference on Machine Learning, {ICML} 2023, 23-29 July
                  2023, Honolulu, Hawaii, {USA}},
  booktitle    = {International Conference on Machine Learning, {ICML} 2023},
  series       = {Proc.\ Machine Learning Res.},
  volume       = {202},
  pages        = {2092--2112},
  publisher    = {{PMLR}},
  year         = {2023},
  url          = {https://proceedings.mlr.press/v202/bennouna23a.html},
  bibsource    = {dblp computer science bibliography, https://dblp.org}
}

@inproceedings{BartheEA12,
  author       = {Gilles Barthe and
                  Boris K{\"{o}}pf and
                  Federico Olmedo and
                  Santiago Zanella{-}B{\'{e}}guelin},
  editor       = {John Field and
                  Michael Hicks},
  title        = {Probabilistic relational reasoning for differential privacy},
  ALTbooktitle    = {Proceedings of the 39th {ACM} {SIGPLAN-SIGACT} Symposium on Principles
                  of Programming Languages, {POPL} 2012, Philadelphia, Pennsylvania,
                  USA, January 22-28, 2012},
  booktitle    = {Principles
                  of Programming Languages, {POPL} 2012},
  pages        = {97--110},
  publisher    = {{ACM}},
  year         = {2012},
  nourl          = {https://doi.org/10.1145/2103656.2103670},
  doi          = {10.1145/2103656.2103670},
  bibsource    = {dblp computer science bibliography, https://dblp.org}
}

@inproceedings{DiniEA13,
  author       = {Gianluca Dini and
                  Fabio Martinelli and
                  Ilaria Matteucci and
                  Andrea Saracino and
                  Daniele Sgandurra},
  editor       = {Joaqu{\'{\i}}n Garc{\'{\i}}a{-}Alfaro and
                  Georgios V. Lioudakis and
                  Nora Cuppens{-}Boulahia and
                  Simon N. Foley and
                  William M. Fitzgerald},
  title        = {Introducing Probabilities in Contract-Based Approaches for Mobile
                  Application Security},
  ALTbooktitle    = {Data Privacy Management and Autonomous Spontaneous Security - 8th
                  International Workshop, {DPM} 2013, and 6th International Workshop,
                  {SETOP} 2013, Egham, UK, September 12-13, 2013, Revised Selected Papers},
  booktitle    = {Data Privacy Management and Autonomous Spontaneous
                  Security, {DPM-SETOP} 2013},
  series       = lncs,
  volume       = {8247},
  pages        = {284--299},
  publisher    = {Springer},
  year         = {2013},
  nourl          = {https://doi.org/10.1007/978-3-642-54568-9\_18},
  doi          = {10.1007/978-3-642-54568-9\_18},
  bibsource    = {dblp computer science bibliography, https://dblp.org}
}

@InProceedings{fswbgkm:quantitative-graded-semantics,
  author = 	 {Jonas Forster and Lutz Schr{\"o}der and Paul Wild
                  and Harsh Beohar and Sebastian Gurke and Barbara
                  K{\"o}nig and Karla Messing}, 
  title = 	 {Quantitative Graded Semantics and Spectra of Behavioural Metrics},
  booktitle =    {Computer Science Logic, {CSL} 2025},
  publisher =	 {Schloss Dagstuhl -- Leibniz Center for Informatics},
  series = 	 {{LIPIcs}},
  year =         2025,
  volume =	 {326},
  pages =	 {33:1--33:21},
  url =          {https://doi.org/10.4230/LIPIcs.CSL.2025.33}
}

@article{MartiVenema15,
  author       = {Johannes Marti and
                  Yde Venema},
  title        = {Lax extensions of coalgebra functors and their logic},
  journal      = {J.\ Comput.\ Syst.\ Sci.},
  volume       = {81},
  number       = {5},
  pages        = {880--900},
  year         = {2015},
  nourl          = {https://doi.org/10.1016/j.jcss.2014.12.006},
  doi          = {10.1016/J.JCSS.2014.12.006},
  bibsource    = {dblp computer science bibliography, https://dblp.org}
}

@PhDThesis{	  Thijs96,
  author	= {Albert Thijs},
  title		= {Simulation and fixpoint semantics},
  school	= {University of Groningen},
  year		= {1996}
}

@InProceedings{	  BackhouseBruinEtAl91,
  author	= {Roland Carl Backhouse and Peter J. de Bruin and Paul F.
		  Hoogendijk and Grant Malcolm and Ed Voermans and Jaap van
		  der Woude},
  editor	= {Maurice Nivat and Charles Rattray and Teodor Rus and
		  Giuseppe Scollo},
  title		= {Polynomial Relators (Extended Abstract)},
  altbooktitle	= {Algebraic Methodology and Software Technology {(AMAST}
		  '91), Proceedings of the Second International Conference on
		  Methodology and Software Technology},
  booktitle	= {Algebraic Methodology and Software Technology, AMAST
		  1991},
  series	= {Workshops in Computing},
  pages		= {303--326},
  publisher	= springer,
  year		= {1991}
}

@InProceedings{	  Levy11,
  author	= {Paul Blain Levy},
  editor	= {Martin Hofmann},
  title		= {Similarity Quotients as Final Coalgebras},
  altbooktitle	= {Foundations of Software Science and Computational
		  Structures - 14th International Conference, {FOSSACS} 2011,
		  Held as Part of the Joint European Conferences on Theory
		  and Practice of Software, {ETAPS} 2011, Saarbr{\"{u}}cken,
		  Germany, March 26-April 3, 2011. Proceedings},
  booktitle	= {Foundations of Software Science and Computational
		  Structures, {FOSSACS} 2011},
  series	= lncs,
  volume	= {6604},
  pages		= {27--41},
  publisher	= springer,
  year		= {2011},
  doi		= {10.1007/978-3-642-19805-2\_3}
}

@inproceedings{GorinSchrode13,
  author       = {Daniel Gor{\'{\i}}n and
                  Lutz Schr{\"{o}}der},
  editor       = {Reiko Heckel and
                  Stefan Milius},
  title        = {Simulations and Bisimulations for Coalgebraic Modal Logics},
  ALTbooktitle    = {Algebra and Coalgebra in Computer Science - 5th International Conference,
                  {CALCO} 2013, Warsaw, Poland, September 3-6, 2013. Proceedings},
  booktitle    = {Algebra and Coalgebra in Computer Science,
                  {CALCO} 2013},
  series       = lncs,
  volume       = {8089},
  pages        = {253--266},
  publisher    = springer,
  year         = {2013},
  nourl          = {https://doi.org/10.1007/978-3-642-40206-7\_19},
  doi          = {10.1007/978-3-642-40206-7\_19},
  bibsource    = {dblp computer science bibliography, https://dblp.org}
}

@article{LarsenSkou91,
  author       = {Kim Guldstrand Larsen and
                  Arne Skou},
  title        = {Bisimulation through Probabilistic Testing},
  journal      = {Inf.\ Comput.},
  volume       = {94},
  number       = {1},
  pages        = {1--28},
  year         = {1991},
  nourl          = {https://doi.org/10.1016/0890-5401(91)90030-6},
  doi          = {10.1016/0890-5401(91)90030-6},
  bibsource    = {dblp computer science bibliography, https://dblp.org}
}

@InProceedings{wspk:van-benthem-fuzzy,
  author       = {Paul Wild and
                  Lutz Schr{\"{o}}der and
                  Dirk Pattinson and
                  Barbara K{\"{o}}nig},
  editor       = {Anuj Dawar and
                  Erich Gr{\"{a}}del},
  title        = {A van Benthem Theorem for Fuzzy Modal Logic},
  ALTbooktitle    = {Proceedings of the 33rd Annual {ACM/IEEE} Symposium on Logic in Computer
                  Science, {LICS} 2018, Oxford, UK, July 09-12, 2018},
  booktitle    = {Logic in Computer
                  Science, {LICS} 2018},
  pages        = {909--918},
  publisher    = {{ACM}},
  year         = {2018},
  nourl          = {https://doi.org/10.1145/3209108.3209180},
  doi          = {10.1145/3209108.3209180},
  bibsource    = {dblp computer science bibliography, https://dblp.org}
}

@inproceedings{GebhartEA25,
  author       = {Stefan Gebhart and
                  Lutz Schr{\"{o}}der and
                  Paul Wild},
  title        = {Non-expansive Fuzzy {$\mathcal{ALC}$}},
  ALTbooktitle    = {Proceedings of the Thirty-Fourth International Joint Conference on
                  Artificial Intelligence, {IJCAI} 2025, Montreal, Canada, August 16-22,
                  2025},
  booktitle    = {International Joint Conference on
                  Artificial Intelligence, {IJCAI} 2025},
  pages        = {4509--4517},
  publisher    = {ijcai.org},
  year         = {2025},
  nourl          = {https://doi.org/10.24963/ijcai.2025/502},
  doi          = {10.24963/IJCAI.2025/502},
  bibsource    = {dblp computer science bibliography, https://dblp.org}
}

@Book{AHS90,
  title      = {Abstract and concrete categories: {T}he joy of cats},
  publisher  = {John Wiley \& Sons Inc.},
  year       = 1990,
  author     = {Ad{\'a}mek, Ji{\v{r}\'i} and Herrlich, Horst and Strecker, George E.},
  noseries     = {Pure and Applied Mathematics (New York)},
  noaddress    = {New York},
  isbn       = {0-471-60922-6},
  note       = {Republished in: Reprints in Theory and Applications of Categories, No. 17 (2006) pp.~1--507},
  mrclass    = {18-02 (08C05 18A20 18B30 54B30)},
  mrnumber   = {MR1051419 (91h:18001)},
  mrreviewer = {Saunders Mac Lane},
  pagetotal  = {xiv + 482},
  url        = {http://tac.mta.ca/tac/reprints/articles/17/tr17abs.html},
}

@article{Barr93,
  author       = {Michael Barr},
  title        = {Terminal Coalgebras in Well-Founded Set Theory},
  journal      = {Theor.\ Comput.\ Sci.},
  volume       = {114},
  number       = {2},
  pages        = {299--315},
  year         = {1993},
  nourl          = {https://doi.org/10.1016/0304-3975(93)90076-6},
  doi          = {10.1016/0304-3975(93)90076-6},
  bibsource    = {dblp computer science bibliography, https://dblp.org}
}

@inproceedings{Straccia98,
  author       = {Umberto Straccia},
  editor       = {Jack Mostow and
                  Chuck Rich},
  title        = {A Fuzzy Description Logic},
  ALTbooktitle    = {Proceedings of the Fifteenth National Conference on Artificial Intelligence
                  and Tenth Innovative Applications of Artificial Intelligence Conference,
                  {AAAI} 98, {IAAI} 98, July 26-30, 1998, Madison, Wisconsin, {USA}},
  booktitle    = { Artificial Intelligence/ Innovative Applications of
                  Artificial Intelligence, {AAAI} 1998,
                  {IAAI} 1998},
  pages        = {594--599},
  publisher    = {{AAAI} Press / The {MIT} Press},
  year         = {1998},
  url          = {http://www.aaai.org/Library/AAAI/1998/aaai98-084.php},
  bibsource    = {dblp computer science bibliography, https://dblp.org}
}

\newpage
\appendix

\section{Additional Details and Omitted Proofs}

\subsection*{Details for \Cref{rem:2-to-2}}

Observe first that for $\epsilon<1$, $\laxtwo_{\epsilon,\Lambda}$ is the
one-sided Egli-Milner lifting, given for $S\in\pfun X$, $T\in\pfun Y$,
and $r\subseteq X\times Y$ by $S \mathrel{\laxtwo_{\epsilon,\Lambda}r} T$
iff for all $x\in S$, there exists $y\in T$ such that
$x\mathrel{r} y$, i.e.\ iff $S\subseteq \rev r[T]$. This is seen as
follows: By definition and by two-valuedness of~$\Diamond$,
$S \mathrel{\laxtwo_{\epsilon,\Lambda}r} T$ iff whenever $S\in\Diamond_X(A)$
for $A\in 2^X$, then $T\in\Diamond_Y(r[A])$. To conclude the previous
condition from this definition, just apply the definition to $A=\{x\}$
for a given $x\in S$. For the converse direction, let $x\in S\cap
A$. Then $x\mathrel{r} y$ for some $y\in T$, and then
$y\in T\cap r[A]$, so $T\in\Diamond_Y(r[A])$.

Thus, we have to show that
\begin{equation*}
  \bigwedge\{\epsilon\mid S\subseteq \rev{r_\epsilon}[T]\}=\bigvee_{x\in S}\bigwedge_{y\in T}r(x,y)
\end{equation*}
for $S\in\pfun X$, $T\in\pfun Y$, $r\colon\frel{X}{Y}$. We split this
equality into two inequalities:

`$\ge$': Let $S\subseteq\rev{r_\epsilon}[T]$, $x\in X$; we have to
show that $\epsilon\ge \bigwedge_{y\in T}r(x,y)$. By hypothesis, there
exists $y\in T$ such that $x\mathrel{r_\epsilon}y$,
i.e.~$r(x,y)\le\epsilon$; the claim follows immediately.

`$\le$': Let $\epsilon=\bigvee_{x\in S}\bigwedge_{y\in T}r(x,y)$, and
let $\delta>0$. It suffices to show that
$S\subseteq \rev{r_{\epsilon'}}[T]$ for
$\epsilon'=\epsilon+\delta$. So let $x\in S$. Then
$\bigwedge_{y\in T}r(x,y)\le\epsilon$, so there exists $y\in T$ such
that $r(x,y)\le\epsilon+\delta=\epsilon'$, so
$x\in \rev{r_\epsilon}[T]$ as required.

\subsection*{Details for \Cref{expl:poly-solvable}}

\begin{enumerate}[wide]
\item \emph{L\'evy-Prokhorov behavioural distance}: The flow network
  $\mathcal{N}=\mathcal{N}(\mu,\nu,R)$ has the set
  $\{\bot,\top\}\cup X\cup Y$ of nodes, where~$\bot$ is the
  source,~$\top$ is the sink, and~$X$ and~$Y$ are w.l.o.g.~assumed to
  be disjoint. Edge capacities are $c(\bot,x)=\mu(x)$ for $x\in X$,
  $c(y,\top)=\nu(y)$ for $y\in Y$, $c(x,y)=1$ for $(x,y)\in R$,
  and~$0$ otherwise. Let $(U,V)$ be a minimum cut
  of~$\mathcal{N}$. (As noted in the main body of the paper, the
  maximum flow of~$\mathcal{N}$ can be computed in polynomial time by
  standard algorithms; the left part~$U$ of the minimum cut then
  consists of all nodes reachable from the source via edges with
  positive residual capacity.) It is easy to see
  (cf.~\cite{DesharnaisEA08}) that $(U,V)$ does not cut any edges
  $(x,y)\in R$; in particular,~$x\in U$ and $(x,y)\in R$ implies
  $y\in U$. Now put
  \begin{equation*}
    C=V\cap X\qquad A=X\setminus C\qquad B=U\cap Y.
  \end{equation*}
  By the above, the capacity~$c$ of $(U,V)$ is
  $\mu(C)+\nu(B)=\mu(X)-\mu(A)+\nu(B)$. Now suppose that the maximum flow
  of~$\mathcal{N}$, which by the min-flow-max-cut theorem equals~$c$,
  is less than $\mu(X)-\epsilon$. Then
  $\nu(B)<(\mu(X)-\epsilon)-(\mu(X)-\mu(A))=\mu(A)-\epsilon$. To show that
  $A\times(Y\setminus B)\subseteq S$, let
  $(x,y)\in A\times (Y\setminus B)$. Assume for a contradiction that
  $(x,y)\in R$. Since $x\in A$, we have $x\in U$, hence by the above
  $y\in U$, and thus $y\in B$, contradiction.
\end{enumerate}

\subsection*{Proof of \Cref{lem:sat-epsilon}}

\begin{enumerate}[wide]
\item\label{item:epsilon-up} Trivial.
\item One shows by induction on~$\phi$ that whenever
  $x\models_\delta\phi$, then $y\models_{\delta+\epsilon}\phi$. The
  only slightly non-obvious case in the induction is disjunction
  $\phi_1\lor\phi_2$. By Item~\ref{item:epsilon-up}, the set of
  indices $i\in\{1,2\}$ such that $y\models_{\delta+\epsilon'}\phi_i$
  decreases monotonically as $\epsilon'$ approaches $\epsilon$ from
  above, and hence becomes stationary (and non-empty) from some
  $\epsilon'$ onwards, so we can apply the induction hypothesis for
  some index in the stationary set.
\end{enumerate}

\subsection*{Proof of \Cref{thm:hm-two-valued}}

  `Only if' is by invariance (Lemma~\ref{lem:invariance}); we show
  `if'. To this end, we show that the relation~$R$ defined by
  \begin{equation*}
    x\mathrel{R}y\quad\text{iff}\quad\text{$y$ logically $\epsilon$-simulates $x$}
  \end{equation*}
  is an $\epsilon$-$\Lambda$-simulation. So let $x\mathrel{R}y$. By
  finite branching, there are finite sets $X'\subseteq X$,
  $Y'\subseteq Y$ such that $\xi(x)\in FX'$,
  $\zeta(y)\in FY'$\lsnote{Introduce subset convention}. Let
  $A\in 2^X$ and put $q=\lambda(A)(\xi(x))$; by naturality
  of~$\lambda$, we can assume that $A\subseteq X'$ (in particular,~$A$
  is finite). We have to show that
  $\lambda(R[A])(\zeta(y))\ge q-\epsilon$. Assume the contrary.

  For all $x\in A$, $y'\not\in R[A]$ there exists a formula
  $\phi_{x,y'}$ such that
  \begin{equation*}
    x\models_0 \phi_{x,y'}\quad\text{but}\quad y'\not\models_{\epsilon}\phi_{x,y'}.
  \end{equation*}
  Define
  \begin{equation*}\textstyle
    \phi=\bigvee_{x\in A}\bigwedge_{y'\not\in R[A], y'\in Y'}\phi_{x,y'}.
  \end{equation*}
  Then, by construction, $A\subseteq\Sem{\phi}_0$, so
  $x\models_0\lambda_q\phi$. Since $x\mathrel{R}y$, this implies
  $y\models_\epsilon\lambda_q\phi$, i.e.\
  $\lambda(\Sem{\phi}_\epsilon)(\zeta(y))\ge q-\epsilon$. Since
  $\lambda(\Sem{\phi}_\epsilon)(\zeta(y))=\lambda(\Sem{\phi}_\epsilon\cap
  Y')(\zeta(y))$ by naturality of~$\lambda$, we thus have -- using the
  fact that $\lambda$ is monotone -- that
  $\Sem{\phi}_\epsilon\cap Y'\not\subseteq R[A]$. But then there
  exists $y'$ with $y'\in \Sem{\phi}_\epsilon$, $y'\in Y'$ and
  $y'\not\in R[A]$. Hence $y'\models_\epsilon\phi$, which implies
  $y'\models_\epsilon\phi_{x,y'}$ for some $x\in A$, a contradiction
  since $y'\not\in R[A]$.


\subsection*{Details for~\Cref{rem:sugeno-duals}}

First, we have the following lemma, stating that the non-strict inequality featuring in the definition of $f_\epsilon = \{x \mid f(x) \ge \epsilon \}$ inside the Sugeno modalities can be replaced by a strict one:
\begin{lem}\label{lem:sugeno-strict}
  For every $f\in\V^X$ and $a\in FX$,
  \begin{equation*}
    \gen{\lambda}(f)(a) = \bigvee_{\epsilon\ge 0} \epsilon \wedge \lambda(\{x \mid f(x) > \epsilon \})(a).
  \end{equation*}
\end{lem}
\begin{proof}
  We prove the two inequalities separately.
  For `$\ge$', we use that $\{x \mid f(x) > \epsilon\} \subseteq \{x \mid f(x) \ge \epsilon\}$ and hence this holds by monotonicity of $\lambda$.

  For `$\le$', assume $\delta < \gen{\lambda}(f)(a)$.
  Then there exists $\epsilon\ge 0$ such that $\delta < \epsilon \wedge \lambda(\{x \mid f(x) \ge \epsilon \})(a)$.
  Then $\{x \mid f(x) \ge \epsilon \} \subseteq \{x \mid f(x) > \delta \}$ and hence $\delta < \lambda(\{x \mid f(x) \ge \epsilon \})(a) \le \lambda(\{x \mid f(x) > \delta \})(a)$ by monotonicity of $\lambda$.
  We therefore have
  \[ \delta \le \delta \wedge \lambda(\{x \mid f(x) > \delta \})(a) \le \bigvee_{\epsilon\ge 0} \epsilon \wedge \lambda(\{x \mid f(x) > \epsilon \})(a), \]
  finishing the proof.
\end{proof}

\noindent
Using this lemma, we have, for every $f\in\V^X$ and $a\in FX$:
\begin{align*}
  \gen{\overline{\lambda}}(f)(a)
  &= \bigvee_{\epsilon\ge 0} \epsilon \wedge \overline{\lambda}(\{x \mid f(x) > \epsilon\})(a) \\
  &= \bigvee_{\epsilon\ge 0} \epsilon \wedge (1 - \lambda(\{x \mid f(x) \le \epsilon\})(a)) \\
  &= \bigvee_{\epsilon\ge 0} (1-\epsilon) \wedge (1 - \lambda(\{x \mid f(x) \le 1-\epsilon\})(a)) \\
  &= 1 - \bigwedge_{\epsilon\ge 0} \epsilon \vee \lambda(\{x \mid f(x) \le 1-\epsilon\})(a)) \\
  &= 1 - \bigwedge_{\epsilon\ge 0} \epsilon \vee \lambda(\{x \mid 1-f(x) \ge \epsilon\})(a)) \\
  &= \overline{\gen{\lambda}}(f)(a)
\end{align*}

\subsection*{Proof of~\Cref{propn:sugeno-modalities}}

  \begin{enumerate}[wide]
  \item Immediate.
  \item Expand
    $\gen{\Diamond_a}(f)(\mu) = \bigvee_{\epsilon\ge 0} \epsilon
      \land \mu\{(a,x)\mid f(x)\ge \epsilon\}$.
    Then use that
    $\mu\{(a,x)\mid f(x)\ge \epsilon\} = \mu\big|_{\{a\}\times X}(f_\epsilon)$.
  \item We have to show that:
    \[\bigvee\{\epsilon\mid \bigvee_{x\in f_\epsilon}
      g(x) \ge \epsilon\}
      =  \bigvee_{x\in X} (g(x)\land f(x)) \]
    We prove two inequalities, showing first that the right-hand side
    is an upper bound for all elements in the left-hand supremum, then
    the reverse:

    ``$\le$'': let $\bar{\epsilon}$\bknote{Can possibly be simplified
      if we choose the maximal such $\bar{\epsilon}$,
      cf.~\ref{lem:G-lambda-alt}} such that
    $\bigvee_{x\in f_{\bar{\epsilon}}} g(x) \ge \bar{\epsilon}$. We
    choose some $\epsilon'<\bar{\epsilon}$ and observe that
    \[ \bigvee_{x\in f_{\epsilon'}} g(x) \ge \bigvee_{x\in
        f_{\bar{\epsilon}}} g(x) \ge \bar{\epsilon} > \epsilon' \] The
    first inequality holds since
    $f_{\bar{\epsilon}} \subseteq f_{\epsilon'}$. Hence, there exists
    $x'\in X$ such that $x'\in f_{\epsilon'}$ and
    $g(x') \ge \epsilon'$. In particular
    $g(x')\land f(x') \ge \epsilon'$. This holds for every
    $\epsilon'<\bar{\epsilon}$, which implies that
    \[\bar{\epsilon} \le \bigvee_{x\in X} (g(x)\land f(x)). \]

    ``$\ge$'': let $x'\in X$. We define
    $\bar{\epsilon} = g(x')\land f(x')$, which implies
    $x'\in f_{\bar{\epsilon}}$. Furthermore
    \[ \bigvee_{x\in f_{\bar{\epsilon}}} g(x) \ge g(x') \ge
      \bar{\epsilon}. \] Hence
    $\bar{\epsilon}\in \{\epsilon\mid \bigvee_{x\in f_\epsilon}
    g(x) \ge \epsilon\}$ and thus it is below the supremum.
  \item We have to show that:
    \begin{eqnarray*}
      && \bigvee\{\epsilon\mid \bigvee_{x\in f_\epsilon, (b,x)\in S}
      1-d(a,b) \ge \epsilon\}  \\
      & = & \bigvee_{(b,x)\in S} ((1-d(a,b))\land f(x))
    \end{eqnarray*}
    As in the previous case we prove two inequalities:

    ``$\le$'': let $\bar{\epsilon}$ be such that
    $\bigvee_{x\in f_{\bar{\epsilon}}, (b,x)\in S} 1-d(a,b) \ge
    \bar{\epsilon}$. We choose some $\epsilon'<\bar{\epsilon}$ and
    observe that
    \[ \bigvee\limits_{\substack{x\in f_{\epsilon'} \\ (b,x)\in S}}
      1-d(a,b) \ge \bigvee\limits_{\substack{x\in f_{\bar{\epsilon}}\\
          (b,x)\in S}} 1-d(a,b) \ge \bar{\epsilon} > \epsilon' \] The
    first inequality holds since
    $f_{\bar{\epsilon}} \subseteq f_{\epsilon'}$. Hence, there exists
    $(b',x')\in S$ such that $x'\in f_{\epsilon'}$ and
    $1-d(a,b') \ge \epsilon'$. In particular
    $(1-d(a,b'))\land f(x') \ge \epsilon'$. This holds for every
    $\epsilon'<\bar{\epsilon}$, which implies that
    \[\bar{\epsilon} \le \bigvee_{(b,x)\in S} ((1-d(a,b))\land f(x)). \]

    ``$\ge$'': let $(b',x')\in S$. We define
    $\bar{\epsilon} = (1-d(a,b'))\land f(x')$ and observe that this implies
    $1-d(a,b') \ge \bar{\epsilon}$ and $f(x') \ge \bar{\epsilon}$, hence
    $x'\in f_{\bar{\epsilon}}$. Furthermore
    \[ \bigvee_{x\in f_\epsilon (b,x)\in S} 1-d(a,b) \ge 1-d(a,b') \ge
      \bar{\epsilon}. \] Hence
    $\bar{\epsilon}\in \{\epsilon\mid \bigvee_{x\in f_\epsilon, (b,x)\in S}
    1-d(a,b) \ge \epsilon\}$ and thus it is below the supremum.


  \item We have
    \begin{align*}
      \gen{\Diamond}(f)(V)
      &= \bigvee_{\epsilon\ge 0} \epsilon\wedge\bigvee_{\mu\in V} \mu(f_\epsilon) \\
      &= \bigvee_{\mu\in V} \bigvee_{\epsilon\ge 0}
      \epsilon\wedge\mu(f_\epsilon) \\
      &= \bigvee_{\mu\in V} \gen{P}(f)(\mu). \qedhere
    \end{align*}
  \end{enumerate}

\subsection*{Proof of~\Cref{thm:expr-sugeno}}

As indicated in the main text, we prove only the harder inequality
$\logdist_{\Gen{\Lambda}}\ge \disteps{\Lambda}$.
  To this end,
  we show that the relation~$R$ defined by
  \begin{equation*}
    c\mathrel{R}d\quad\text{iff}\quad \logdist_{\Gen{\Lambda}}(c,d)\le\epsilon
  \end{equation*}
  is an $\epsilon$-$\Lambda$-simulation. So let $c\mathrel{R}d$. By
  finite branching, there are finite sets $X'\subseteq X$,
  $Y'\subseteq Y$ such that $\xi(c)\in FX'$,
  $\zeta(d)\in FY'$\lsnote{Introduce subset convention}. Let
  $A\in 2^X$ and put $q=\lambda(A)(\xi(c))$; by naturality
  of~$\lambda$, we can assume that $A\subseteq X'$ (in particular,~$A$
  is finite). We have to show that
  $\lambda(R[A])(\zeta(d))\ge q-\epsilon$. Assume the contrary. Then
  for all elements~$B$ of the finite set
  \begin{equation*}
    \FA:=\{B\subseteq Y'\mid \lambda(B)(\zeta(d))\ge q-\epsilon\},
  \end{equation*}
  we have $B\not\subseteq R[A] $ by monotonicity of~$\lambda$, so we
  can fix $y_B\in B\setminus R[A]$.  This means that for all $x\in A$,
  $x\not\mathrel{R}y_B$; we thus have a formula $\phi_{x,B}$ such that
  \begin{equation*}
    \Sem{\phi_{x,B}}(y_B)<\Sem{\phi_{x,B}}(x)-\epsilon.
  \end{equation*}
  Thanks to constant shifts\bknote{Shifts with reals or rationals?},
  we can assume that $\Sem{\phi_{x,B}}(x)=q$ for all $x\in A$,
  $B\in\FA$. Put
  \begin{equation*}\textstyle
    \phi=\bigvee_{x\in A}\bigwedge_{B\in\FA}\phi_{x,B}.
  \end{equation*}
  Then by construction, $A\subseteq\Sem{\phi}_q$, so
  $\lambda(\Sem{\phi}_q)(\xi(c))\ge \lambda(A)(\xi(c))= q$; by
  \Cref{lem:G-lambda-alt}, we thus obtain
  $\Sem{\gen{\lambda}\phi}(c)=\gen{\lambda}\Sem{\phi}(\xi(c))\ge
  q$. Since $c\mathrel{R}d$, this implies
  $\gen{\lambda}_{Y}(\Sem{\phi})(\zeta(d))=\Sem{\gen{\lambda}\phi}(d)\ge
  q-\epsilon$. By naturality,
  $\gen{\lambda}_{Y}(\Sem{\phi})(\zeta(d))=\gen{\lambda}_{Y'}(\Sem{\phi}|_{Y'})(\zeta(d))$\lsnote{Maybe
    introduce restriction notation}; by \Cref{lem:G-lambda-max} and
  finiteness of~$Y'$, it follows that
  $\lambda_{Y'}(\Sem{\phi}_{q-\epsilon}\cap Y')(\zeta(d))\ge
  q-\epsilon$. We thus have $B:=\Sem{\phi}_{q-\epsilon}\cap
  Y'\in\FA$. Thus,
  $\Sem{\phi_{x,B}}(y_B)<\Sem{\phi_{x,B}}(x)-\epsilon=q-\epsilon$ for
  all $x\in A$, and hence $\Sem{\phi}(y_B)<q-\epsilon$. On the other
  hand, $y_B\in B\subseteq \Sem{\phi}_{q-\epsilon}$, contradiction. \qed

\subsection*{Proof of~\Cref{thm:extract-quant}}

  The complexity estimates are immediate. We prove the correctness
  claim by induction on the stage~$i$, i.e.~we show that~$\phi$ as
  per~\eqref{eq:distinguishing-phi} is an $\epsilon$-distinguishing
  formula for $(x_0,y_0)$. By the normation step, we have
  $\phi(x)\ge\bigwedge_{y\in Y\setminus B}\phi_{xy}\ge q$ for all
  $x\in A$. That is, $A\subseteq\Sem{\phi}_q$, and therefore
  $\Sem{\gen{\lambda}\phi}(x_0)\ge\lambda(\Sem{\phi}_q)(\xi(x))\wedge
  q\ge\lambda(A)(\xi(x))\wedge q=q$. On the other hand, we claim
  that $\Sem{\gen{\lambda}\phi}(y_0)< q-\epsilon$. By
  \cref{lem:G-lambda-max}, it suffices to show that
  $\lambda(\Sem{\phi})_{q-\epsilon}(\zeta(y_0))<q-\epsilon$. Since
  $(A,B)$ is a legal move for Spoiler at $(x_0,y_0)$, we have
  $\lambda(B)(\zeta(y_0))<q-\epsilon$, so we are done once we show
  that $\Sem{\phi}_{q-\epsilon}\subseteq B$. We show the
  contraposition: Let $y\in Y\setminus B$; we show that
  $\Sem{\phi}(y)<q-\epsilon$. So let $x\in A$; we have to show that
  $\Sem{\bigwedge_{y'\in Y\setminus}\phi_{xy'}}(y)<q-\epsilon$, which
  follows from $\Sem{\phi_{xy}}(y)<q-\epsilon$. \qed

\subsection*{Proof of~\Cref{thm:lp-kantorovich-coalg}}

  Let $r\colon\frel{X}{Y}$ be a $\V$-valued relation, and let
  $a\in FX,b\in FY$. We split the claimed equality
  $\laxv_\Lambda r(a,b)=\kantv_{\Gen{\Lambda}} r(a,b)$ into two inequalities:

  $\laxv_\Lambda r(a,b)\ge \kantv_{\Gen{\Lambda}} r(a,b)$: Let $(f,g)$ be $r$-preserved,
  and let $\lambda\in\Lambda$. We have to show that
  \begin{equation*}
    \gen{\lambda}_X(f)(a)\ominus\gen{\lambda}_Y(g)(b)\le \laxv_\Lambda r(a,b).
  \end{equation*}
  So let $a\mathrel{L_\delta r_\delta}b$; by the definition of
  $\laxv_\Lambda$, we have to show that
  $\gen{\lambda}_X(f)(a)\ominus\gen{\lambda}_Y(g)(b)\le\delta$,
  equivalently that
  \begin{equation*}
    \gen{\lambda}_Y(g)(b)\ge\gen{\lambda}_X(f)(a)-\delta,
  \end{equation*}
  or, after expanding definitions,
  \begin{eqnarray*}\textstyle
    && \bigvee_{\epsilon'\ge 0}\epsilon'\wedge\lambda(g_{\epsilon'})(b)\ge
    \big(\bigvee_{\epsilon\ge
      0}\epsilon\wedge\lambda(f_{\epsilon})(a)\big)-\delta \\
    &=& \bigvee_{\epsilon\ge 0}(\epsilon-\delta)\wedge(\lambda(f_{\epsilon})(a)-\delta).
  \end{eqnarray*}
  We prove this inequality among suprema by showing that every element
  in the set whose supremum is taken on the right on the right is
  majorized by an element on the left. Indeed, let $\epsilon\ge 0$; we
  claim that
  \begin{equation*}
    \epsilon'\wedge \lambda(g_{\epsilon'})(b)\ge(\epsilon-\delta)\wedge(\lambda(f_\epsilon)(a)-\delta)
  \end{equation*}
  for $\epsilon'=\epsilon\ominus\delta$. To prove this, it suffices to
  show that
  \begin{equation*}
    \lambda(g_{\epsilon'})(b)\ge \lambda(f_\epsilon)(a)-\delta.
  \end{equation*}
  Since $a\mathrel{L_\delta r_\delta}b$, we have
  $\lambda(r_\delta[f_\epsilon])(b)\ge\lambda(f_\epsilon)(a)-\delta$, so we
  are done by monotonicity of~$\lambda$ once we show that
  \begin{equation*}
    r_\delta[f_\epsilon]\subseteq g_{\epsilon'}.
  \end{equation*}
  So let $f(x)\ge\epsilon$ and $r(x,y)\le\delta$. We have to show that
  $g(y)\ge\epsilon'$. Indeed, since $(f,g)$ is
  $r$-preserved, we have
  \begin{equation*}
    g(y)\ge f(x)-r(x,y)\ge \epsilon-\delta.
  \end{equation*}
  As $g(y)\ge 0$ holds trivially, we have $g(y) \ge \epsilon\ominus\delta = \epsilon'$.

  $\laxv_\Lambda r(a,b)\le \kantv_{\Gen{\Lambda}} r(a,b)$: Suppose that
  $\kantv_\Lambda r(a,b)\le\delta$. It suffices to show that
  $\laxv_\Lambda r(a,b)\le \delta+u$ for all (sufficiently small) $u>0$,
  i.e.\ that $a\mathrel{L_{\delta+u}r_{\delta+u}}b$. So let
  $A\subseteq X$ and $\lambda\in\Lambda$; we have to show that
  \begin{equation}\label{eq:r-delta-sim}
    \lambda(r_{\delta+u}[A])(b)\ge\lambda(A)(a)-\delta-u.
  \end{equation}
  assuming w.l.o.g.~that $\delta+u\le\lambda(A)(a)$. Put
  $\epsilon_A=\lambda(A)(a)$, and define predicates $f\colon X\to\V$,
  $g\colon Y\to\V$ by
  \begin{align*}\textstyle
    f(x) & =
           \begin{cases}
             \epsilon_A & x\in A\\
             0 & \text{otherwise}
           \end{cases}
           & g(y) & = \bigvee_{x\in X}f(x)\ominus r(x,y).
  \end{align*}
  Then $(f,g)$ is $r$-preserved by construction (indeed,
  $f(x)-g(y)\le f(x)-(f(x)\ominus r(x,y))\le r(x,y)$), and
  $f_{\epsilon_A}=A$. Since $\kantv_\Lambda r(a,b)\le\delta$, we obtain
  \begin{eqnarray*}\textstyle
    && \bigvee_{\epsilon'\ge 0}\epsilon'\wedge\lambda(g_{\epsilon'})(b)\ge
    \big(\bigvee_{\epsilon\ge
      0}\epsilon\wedge\lambda(f_\epsilon)(a)\big)-\delta \\
    &=& \bigvee_{\epsilon\ge 0}(\epsilon-\delta)\wedge(\lambda(A)(a)-\delta).
  \end{eqnarray*}
  In particular, this implies that there exists
  $\epsilon'\ge 0$ such that
  \begin{equation*}
    \epsilon'\wedge\lambda(g_{\epsilon'})(b)\ge
    ((\epsilon_A-\delta)\wedge(\lambda(A)(a)-\delta))-u=\epsilon_A-\delta-u,
  \end{equation*}
  i.e.~$\epsilon'\ge\epsilon_A-\delta-u$ and
  $\lambda(g_{\epsilon'})(b)\ge\epsilon_A-\delta-u$. We can assume that
  $\epsilon'>0$ (this is automatic if $\epsilon_A>\delta+u$; otherwise,
  $\epsilon'$ can be chosen arbitrarily). From
  $\lambda(g_{\epsilon'})(b)\ge\epsilon_A-\delta-u$, our
  goal~\eqref{eq:r-delta-sim} follows by monotonicity of~$\lambda$
  once we show that
  \begin{equation*}
    g_{\epsilon'}\subseteq r_{\delta+u}[A].
  \end{equation*}
  So let~$y\in g_{\epsilon'}$. There exists~$x$ such that
  $f(x)\ominus r(x,y)\ge g(y)-u\ge\epsilon'-u$. Since $\epsilon'>0$,
  we can assume that $\epsilon'-u>0$; then the previous inequality
  means that $f(x)-r(x,y)\ge\epsilon'-u\ge\epsilon_A-\delta-u$. This
  implies $f(x)=\epsilon_A$, so $x\in f_{\epsilon_A}=A$, and moreover
  $r(x,y)\le f(x)-\epsilon_A+\delta+u=\delta+u$, so
  $x\mathrel{r_{\delta+u}}y$, showing that $y\in r_{\delta+u}[A]$ as required. \qed

\end{document}